\documentclass[12pt,letterpaper]{article}

\usepackage[margin=1in]{geometry}

\usepackage{amsmath,amssymb,amsthm}
\usepackage{bm}
\usepackage{graphicx}
\usepackage{xcolor}
\usepackage{booktabs}
\usepackage{siunitx}
\usepackage{xspace}

\definecolor{CBF-orange}{HTML}{FFC107}
\definecolor{CBF-blue}{HTML}{1E88E5}
\definecolor{CBF-green}{HTML}{004D40}
\definecolor{CBF-pink}{HTML}{D81B60}

\usepackage[colorlinks=true,linkcolor=violet,citecolor=violet,urlcolor=violet]{hyperref}

\title{NervePool: A Simplicial Pooling Layer}
\author{Ernst R{\"o}ell \\ Helmholtz Munich and Technical University of Munich, Germany \\ \texttt{ernst.roell@helmholtz-munich.de}
\and Sarah McGuire Scullen \\ Pacific Northwest National Lab, USA \\ \texttt{sarah.scullen@pnnl.gov}
\and Elizabeth Munch \\ Michigan State University, USA \\ \texttt{muncheli@msu.edu}
\and Bastian Rieck \\ University of Fribourg, Switzerland \\ \texttt{bastian.grossenbacher@unifr.ch}
\and Matthew Hirn \\ Michigan State University, USA \\ \texttt{mhirn@msu.edu}}
\date{}

\newcommand{\revision}[1]{#1}

\newcommand{\Ab}{\mathbf{A}}

\newcommand{\Bb}{\mathbf{B}}

\newcommand{\Db}{\mathbf{D}}
\newcommand{\Sb}{\mathbf{S}}
\newcommand{\Xb}{\mathbf{X}}
\newcommand{\Zb}{\mathbf{Z}}

\newcommand{\cU}{\mathcal{U}}

\newcommand{\R}{\mathbb{R}}

\usepackage{bm}
\usepackage{blkarray}
\usepackage{xspace}

\usepackage{siunitx}
\usepackage{booktabs}

\newcommand{\nervePool}{{\scshape{NervePool}}}
\newcommand{\NervePool}{{\scshape{NervePool}}}
\newcommand{\nervepool}{{\scshape{NervePool}}}
\newcommand{\modelname}{{\scshape{NervePool}}}

\newcommand{\St}{\mathrm{St}}
\newcommand{\Nrv}{\mathrm{Nrv}}
\newcommand{\MPSN}{{\mathrm{MPSN}}}
\newcommand{\softmax}{\mathrm{softmax}}

\newtheorem{theorem}{Theorem}

\usepackage{enumitem}
\newenvironment{romanenumerate}{\begin{enumerate}[label=(\roman*)]}{\end{enumerate}}

\begin{document}
\maketitle

\begin{abstract}
	For deep learning problems on graph-structured data,
	pooling layers are important for down sampling, reducing
	computational cost, and to minimize overfitting.
We define a pooling layer, \nervePool{}, for data
	structured as simplicial complexes, which are
	generalizations of graphs that include higher-dimensional
	simplices beyond vertices and edges; this structure allows
	for greater flexibility in modeling higher-order relationships.
The proposed simplicial coarsening scheme is built upon
	partitions of vertices, which allow us to generate
	hierarchical representations of simplicial complexes,
	collapsing information in a learned fashion.
\nervePool{} builds on the learned vertex cluster
	assignments and extends to coarsening of higher dimensional
	simplices in a deterministic fashion.
While in practice the pooling operations are computed via
	a series of matrix operations, the topological motivation
	is a set-theoretic construction based on unions of stars of
	simplices and the nerve complex.
\end{abstract}
 \section{Introduction}
Design of deep learning architectures for tasks on spaces of
topological objects has seen rapid development, fueled by the
desire to model higher-order interactions that are naturally
occurring in data.
In this topological framework, we can consider data
structured as simplicial complexes, generalizations of graphs
that include higher-dimensional simplices beyond vertices and edges;
this structure allows for greater flexibility in modeling
higher-order relationships.
Concepts from graph signal processing have been generalized
for this higher-order network setting, leveraging operators
such as Laplacian matrices
\cite{Schaub2021SigProcessingHigherOrderNetworks,
	roddenberry2021cellcomplexsignalproc, sardellitti2022topological}.
Subsequently, there have been many developments regarding
deep learning architectures that leverage simplicial (and
cell) complex structure.
These methods have been designed through the lens of both
convolutional neural networks (CNNs)
\cite{ebli2020SNN,yang2021SCNN,bunch2020simplicial,yang2023SCCNN,
	roddenberry21a, Keros2022Dist2Cycle} and message passing
neural networks \cite{bodnar2021simplicial, bodnar2021CW}.
For the purposes of this paper, we refer to neural networks
within this general family as simplicial neural networks, and
note that our work could be adapted and used for pooling
within any of these architectures.

Typically, the pooling methods for graphs are generalizations
of standard pooling layers for grids (i.e.~those used in CNNs).
As opposed to CNNs, which rely on a natural notion of spatial
locality in the data to apply convolutions in local sections,
the non-regular structure of graphs makes the spatial
locality of graph pooling less obvious \cite{ying2018diffpool}.
For this reason, many graph pooling layers for graph neural
networks (GNNs) rely on local structural information
represented by adjacency matrices, and coarsening operations
applied directly on local graph neighborhoods.
A recent survey of GNN pooling methods proposed a general
framework for the operations which define different pooling
layers: \emph{selection}, \emph{reduction}, and \emph{connection} (SRC)
\cite{2021_SRCgraphpooling}.
Selection refers to the method in which vertices of the input
graph are grouped into clusters; reduction computation is the
aggregation of vertices into meta-vertices (and
correspondingly aggregating features on vertices); connection
refers to determining adjacency information of the
meta-vertices and outputting the pooled graph.
This broad categorization of graph pooling methods into three
computational steps provides a useful framework which can be
extended in a natural way to describe simplicial complex pooling.

In the simplicial complex domain, the notion of spatial
locality necessary for pooling is further complicated due to
the inclusion of higher-dimensional simplices.
This task of coarsening simplicial complexes requires
additional considerations to those of the graph coarsening
task, primarily to ensure that the pooled representation
upholds the definition of a simplicial complex.
Naturally, simplicial complex coarsening shares some
challenges with the task of coarsening graph structured data:
lack of inherent spatial locality and differing input sizes
(varying numbers of nodes and edges).
However, the additional challenge for coarsening in the
simplicial complex setting is the addition of higher
dimensional simplices, with face (and coface) relations in
both dimension directions by the definition of a simplicial complex.
There is also a notable computational challenge when dealing
with simplicial complexes due to the inherent size expansion
with the addition of higher-dimensional simplices.
A simplicial complex with all possible simplices included is
essentially the power set of its vertices.
Thus, controlling the computational explosion when using
simplicial complexes in deep learning frameworks motivates
the use of a pooling layer defined on the space of simplicial complexes.
These pooling layers can be interleaved with regular
simplicial neural network layers in order to reduce the size
of the simplicial complex, which reduces computational
complexity of the model and can also limit overfitting.

Existing graph pooling approaches include \emph{cluster-based methods}
(e.g.\ DiffPool \cite{ying2018diffpool}, MinCutPool
\cite{bianchi2020spectral}), i.e., methods based on learning soft
cluster assignments for vertices to provide new graph-level outputs at
multiple scales, and \emph{sorting methods} (e.g.\ TopK
Pool
\cite{gao2019graphUnets,cangea2018sparsegraph,knyazev2019understanding},
SAGpool \cite{lee2019self}), i.e., methods that rank vertices and
provide coarsening via induced subgraphs.
There are also topologically motivated graph pooling methods
such as Structural Deep Graph Mapper \cite{bodnar2021mapper},
which is based on soft cluster assignments.
It uses differentiable and fixed PageRank-based lens
functions for the standard Mapper algorithm \cite{singh2007_mapper}.
Additionally, there is a method for graph pooling which uses
maximal cliques \cite{luzhnica2019clique}, assigning vertices
to meta-vertices using topological information encoded in the
graph structure.
However, if directly generalized for simplicial complexes,
this clique-based coarsening method never retains any higher
dimensional simplices since the pooled output is always a graph.
Recently, general strategies for pooling simplicial complexes
were proposed \cite{cinque2022poolingSCN}, which directly
generalize different graph pooling methods to act on
simplicial complexes.
This framework aims to generalize graph pooling
methods for simplicial complexes, albeit using different
tools~(i.e., signal processing) from what we use here.
The reader interested in understanding how pooling approaches can be
unified and harmonized using the language of combinatorial complexes is recommended to refer to Hajij et
al.~\cite{Hajij2022} or Papamarkou et al.~\cite{Papamarkou2024}, with
the latter providing an introduction to the nascent field of
\emph{topological deep learning} into which we would contextualize our
proposed method.

In this paper we introduce \NervePool{}, which extends
existing methods for graph pooling to be defined on the full
simplicial complex using standard tools from combinatorial topology.
While in practice, the pooling operations are computed via a
series of matrix operations, \nervePool{} has a compatible
topological formulation which is based on unions of stars of
simplices and a nerve complex construction.
Like graph pooling methods, the \NervePool{}  layer for
simplicial complexes can also be categorized by the SRC graph
pooling framework \cite{2021_SRCgraphpooling}, with necessary
extensions for higher dimensional simplices.
The main contributions of this work are as follows:
\begin{itemize}
	\item We propose a learned coarsening method for simplicial
	      complexes called \nervePool{}, which can be used within
	      neural networks defined on the space of simplicial
	      complexes, for any standard choice of graph pooling on the vertices.
\item Under the assumption of a hard partition of vertices,
	      we prove that \nervePool{} maintains invariance
	      properties of pooling layers necessary for simplicial
	      complex pooling in neural networks, as well as additional
	      properties to maintain simplicial complex structure after pooling.
\end{itemize}

Our implementation and full experimental setup is available at
\url{https://github.com/aidos-lab/nerve-pool} under a BSD-3-Clause license.

 \section{Background}
\label{sec:background}
This section gives background context for some important
concepts in topological data analysis (TDA) and deep learning
methods for data modeled as simplicial complexes.
We first introduce relevant background on simplicial
complexes, chain complexes, boundary maps between linear
spaces generated by $p$-simplices, and simplicial maps.
Additionally, we outline different notions of local
neighborhoods on simplicial complexes through adjacency and
coadjacency relations.

\subsection{Simplicial Complexes}\label{sec:SimplicialComplexes}
A simplicial complex is a generalization of a graph or network.
While graphs can model relational information between pairs
of objects (via vertices and edges connecting them),
simplicial complexes are able to model higher-order interactions.
With this construction, we can still model pair-wise
interactions as edges, but also triple interactions as
triangles, four-way interactions as tetrahedra, and so on.

The building blocks of simplicial complexes are
\textit{$p$-dimensional simplices} (or $p$-simplices for
short), formally defined as a set of $p+1$ vertices $v_i \in V$,
\begin{equation*}
	\sigma_p = (v_0,v_1,\hdots,v_p )\, ,
\end{equation*}
where $V$ is a non-empty vertex set.
Note that we often use a subscript on the simplex to denote
the dimension of the simplex, i.e.~$\dim(\sigma_p) = p$.
The cyclic ordering of the vertices of a simplex 
induces an orientation of that simplex, which can be fixed if necessary for the given task.

Simplicial complexes are collections of these $p$-simplices,
glued together with the constraint that they are closed under
taking subsets.
In other words, an \textit{abstract simplicial complex}, $K$,
is a finite collection of non-empty subsets of $V$ such that
for a simplex $\alpha \in K$, $\beta \subseteq \alpha$
implies $\beta \in K$.
In this case we call $\beta$ a face of $\alpha$ and write
$\beta\leq \alpha$.
Abstract simplicial complexes are combinatorial objects,
defined in terms of collections of vertices ($p$-simplices),
however they have geometric realizations corresponding to
vertices, edges, triangles, tetrahedra, etc.
The dimension of a simplicial complex is defined as the
maximum dimension of all of its simplices:
\begin{equation*}
	\dim(K) = \max_{\sigma \in K} \dim(\sigma)\, ,
\end{equation*}
where the dimension of a simplex $\dim(\sigma)$ is one less
than the cardinality of the vertex set which defines it.
We denote the simplicial complex at layer $\ell$ of a neural
network by $K^{(\ell)}$ and its dimension
$\mathcal{P}^{(\ell)}= \dim(K^{(\ell)})$.
Note that the superscript indexes the neural network layer
and does not represent the $\ell$-skeleton of the complex as
is common in the topology literature.

\paragraph*{Boundary matrices and orientation}
Boundary matrices provide a map between vector spaces
generated by $p$-simplices and $(p-1)$-simplices.
For each dimension $p$, we can define a vector space
$C_{p}(K)$ with a basis given by the set of $p$-simplices and
with its elements given by linear combinations of $p$-simplices.
Such vector spaces are called \textit{chain groups}.
The \textit{boundary map} $\partial_p$ operates on an
oriented simplex $\sigma = [v_0,\cdots, v_p]$ and is defined as:
\begin{equation*}
	\partial_p(\sigma) := \sum_{i=0}^p (-1)^i
	[v_0,\cdots,\widehat{v}_i,\cdots, v_p] \in C_{p-1} (K) \, ,
\end{equation*}
where $\widehat{v}_i$ denotes the removal of vertex $v_i$
from the vertex subset.
Since we are working in a linear space, this definition can
be extended linearly to operate on the entire vector space
$C_p (K)$ generated by  $p$-simplices so that we have a map
$\partial_p : C_p (K) \rightarrow C_{p-1}(K)$ with
\begin{equation*}
	\partial_p (\alpha) := \sum_{i=1}^{n_p} \alpha_i \partial_p
	(\sigma_i) \, .
\end{equation*}

Each boundary matrix captures incidence relations for a given
dimension, keeping track of which simplices of dimension
$p-1$ are faces of a simplex of dimension $p$.
For example, in Fig.~\ref{fig:boundary_matrices}, vertices
$v_0$ and $v_1$ are on the boundary of edge $e_0$ and edge
$e_1$ is on the boundary of face $f_0$.
For each dimension $p$, this incidence-relation information
is encoded in the matrices corresponding to the boundary
operators $\partial_p$.
For clarity in subsequent sections, we equivalently use
notation $\partial_p := \Bb_p \in \R^{n_{p-1} \times n_{p}}$
to represent the matrix boundary operator, with subscripts
indicating the dimension.
When necessary to associate a boundary matrix to a specific
simplicial complex $K$, we denote this matrix by $\Bb_{K,p}$.
Note that $\Bb_0 = \bm{0}$, since there are no simplices of
negative dimension to map down to.

Simplex orientation is encoded in boundary matrices using
$\pm1$, with different orientations of the same simplex
indicated by different placement of negative values in the
boundary matrix.
For each simplex, there are two possible orientations, each
of which are equivalence classes representing all possible
cyclic permutations of the vertices which define the simplex.
We say that $\sigma_{p-1}=[v_0,v_1,\hdots,v_{n_{p-1}}]$ and
$\sigma_{p}=[w_0,w_1,\hdots,w_{n_{p}}]$ have the same
orientation if the ordered set of vertices
$[v_0,v_1,\hdots,v_{n_{p-1}}]$ is contained in any cyclic
permutation of the vertices $w_i$ forming $\sigma_p$.
Conversely, if the ordered set of vertices for $\sigma_{p-1}$
are contained in any cyclic permutations for the other
orientation of $\sigma_p$, then we say that $\sigma_{p-1}$
and $\sigma_p$ have opposite orientation.
Thus, we can fix an orientation and define entries of
oriented $p$-boundary matrices as follows:
\begin{equation*}
	\Bb_p(\sigma_{p-1}, \sigma_p) =
	\begin{cases}
		1  & \text{if } \sigma_{p-1} \text{ and } \sigma_{p}
		\text{ have the same orientation}                                 \\
		-1 & \text{if } \sigma_{p-1} \text{ and } \sigma_{p}
		\text{ have opposite orientation}                                 \\
		0  & \text{if } \sigma_{p-1} \text{ is not a face of } \sigma_{p}
	\end{cases}
\end{equation*}
where $\sigma_p$ and $\sigma_{p-1}$ are codimension-1 simplices in $K$.
If simplex orientation is not necessary, we use the
non-oriented boundary matrix $\left|\Bb_p\right|$.
\begin{figure}
	\centering
	\begin{minipage}[b]{.2\textwidth}
		\raisebox{-0.5\height}{\includegraphics[width=\textwidth]{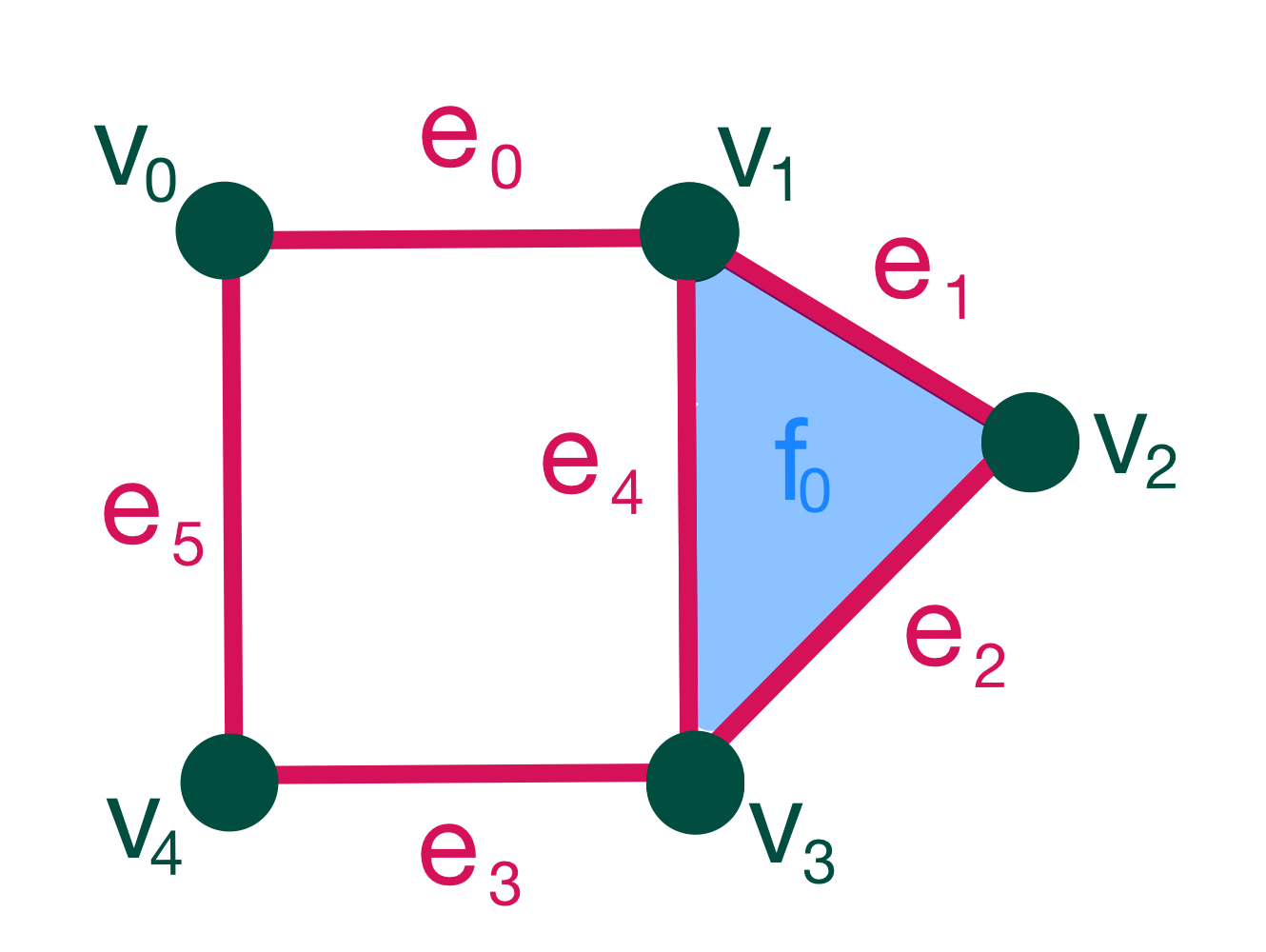}}
	\end{minipage}\begin{minipage}[b]{.8\textwidth}
		\centering
		\small
		\(
		\Bb_1=
		\begin{blockarray}{ccccccc}
			&\textcolor{CBF-pink}{e_0} &\textcolor{CBF-pink}{e_1}
			&
			\textcolor{CBF-pink}{e_2}&\textcolor{CBF-pink}{e_3}&\textcolor{CBF-pink}{e_4}&\textcolor{CBF-pink}{e_5}
			\\
			\begin{block}{c(cccccc)}
				\textcolor{CBF-green}{v_0} & 1 & 0 & 0 & 0 & 0 & 1 \\
				\textcolor{CBF-green}{v_1} & 1 & 1 & 0 & 0 & 1 & 0 \\
				\textcolor{CBF-green}{v_2} & 0 & 1 & 1 & 0 & 0 & 0 \\
				\textcolor{CBF-green}{v_3} & 0 & 0 & 1 & 1 & 1 & 0 \\
				\textcolor{CBF-green}{v_4} & 0 & 0 & 0 & 1 & 0 & 1 \\
			\end{block}
		\end{blockarray}
		\quad
		\Bb_2=
		\begin{blockarray}{cc}
			& \textcolor{CBF-blue}{f_0} \\
			\begin{block}{c(c)}
				\textcolor{CBF-pink}{e_0} & 0 \\
				\textcolor{CBF-pink}{e_1} & 1 \\
				\textcolor{CBF-pink}{e_2} & 1 \\
				\textcolor{CBF-pink}{e_3} & 0 \\
				\textcolor{CBF-pink}{e_4} & 1 \\
				\textcolor{CBF-pink}{e_5} & 0 \\
			\end{block}
		\end{blockarray}
		\)
	\end{minipage}
	\caption{Non-oriented boundary matrices $\Bb_1$ and $\Bb_2$
		for an example simplicial complex.}
	\label{fig:boundary_matrices}
\end{figure}

\paragraph*{Simplicial Maps}A simplicial map defines a function from one simplicial
complex $K_a$ to another $K_b$, specifically by a map from
the vertices of $K_a$ to the vertices of $K_b$ \cite{munkres}.
This map $f: V(K_a) \rightarrow V(K_b)$ must satisfy the following: $\forall \sigma\in K_a, f(\sigma)\in K_b$. 
This way, all images of the vertices of a simplex span a
simplex.

\paragraph*{Adjacency}\label{sec:adjacency}
A useful way to represent a graph is by the corresponding
adjacency matrix: a square matrix that stores information
about which vertices are connected by edges.
However, while simplicial complex structure is also captured
by adjacency relations, there is more than one way to
generalize this idea to higher dimensional simplices.
Specifically, $p$-simplices can relate to each other by
either their common upper or lower dimensional neighbors.
In order to fully capture the face and coface relations in a
simplicial complex, we consider four different adjacency types:
boundary adjacent,
coboundary adjacent,
upper adjacent, and
lower adjacent \cite{bodnar2021simplicial},
which we define next and are shown for an example simplex in
Fig.~\ref{fig:neighbors}.
\begin{figure}
	\centering
	\includegraphics[width=.4\textwidth]{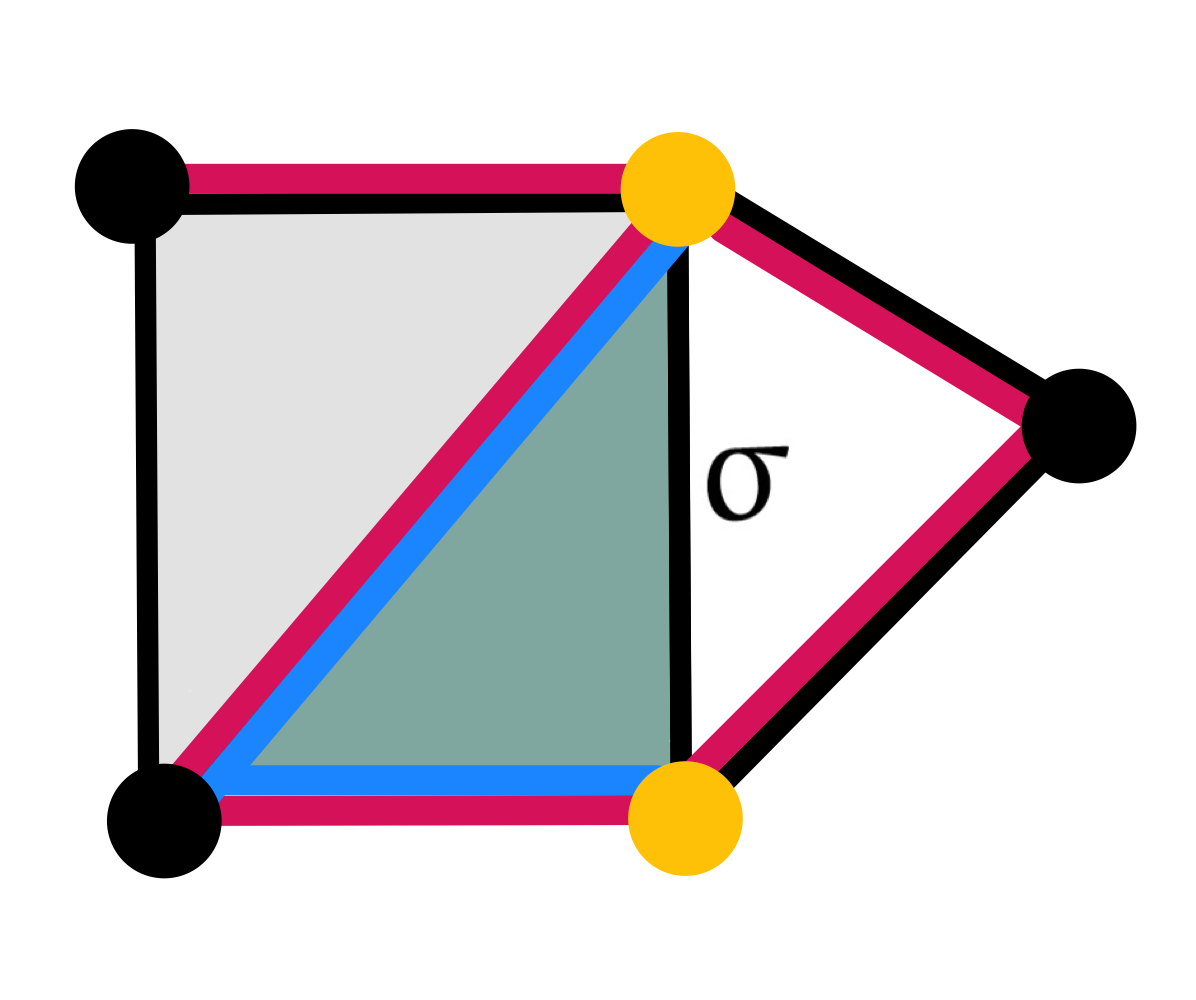}
	\caption{For simplex $\sigma$, geometrically realized as an
		edge, its four different types of adjacent simplices are:
		\textcolor{CBF-orange}{boundary adjacent},
		\textcolor{CBF-green}{coboundary adjacent},
		\textcolor{CBF-pink}{lower adjacent}, and
		\textcolor{CBF-blue}{upper adjacent}.}
	\label{fig:neighbors}
\end{figure}

Fix a $p$-simplex $\sigma_p$.
The \textit{boundary adjacent} simplices are the set of
$(p-1)$-dimensional faces of the $p$-simplex.
This boundary adjacent relation directly corresponds to the
standard boundary map for simplicial complexes: $\Bb_p$.
For the same simplex, the set of $(p+1)$-dimensional
simplices which have $\sigma_p$ as a face are its
\textit{coboundary adjacent} simplices.
The coboundary adjacency relation corresponds to the
transpose of the standard boundary map: $\Bb_{p+1}^T$.

The usual notion of adjacency on a graph corresponds to the
set of vertices which share an edge; i.e.~a higher
dimensional simplex with each as a face.
More generally, for simplicial complexes, $p$-simplices are
considered \textit{upper adjacent} if there exists a
$(p+1)$-simplex that they are both faces of.
We can capture upper adjacent neighbor relations in terms of
the complex's boundary maps (the \textit{up-down
	combinatorial Laplacian operator}):
\begin{equation}\label{eqn:Aup}
	\Ab_{up,p} = \Bb_{p+1}\Bb_{p+1}^T \, .
\end{equation}
The \textit{lower adjacent} neighbors are $p$-simplices such
that there exists a $(p-1)$-simplex that is a face of both
(i.e.~both $p$-simplices are cofaces of a common $(p-1)$-simplex).
The lower adjacent neighbor relations can also be written in
terms of the complex's boundary maps (the \textit{down-up
	combinatorial Laplacian operator}):
\begin{equation*}
	\Ab_{low,p} = \Bb_p^T\Bb_p \, .
\end{equation*}
Note that the sum of these two (upper and lower) adjacency
operators gives the \textit{$p$-dimensional Hodge Laplacian},
which is leveraged in different simplicial neural network
architectures to facilitate local information sharing on complexes.
Using the Hodge Laplacian as a diffusion operator, there are
various existing neural network architectures (e.g.
\cite{bodnar2021simplicial, ebli2020SNN}) that operate on
simplicial complexes.
Simplicial neural networks of this type are the context in
which our pooling layer for simplicial complexes,
\nervePool{}, can be leveraged.

 \section{Methods}
\label{sec:method}
In this section, we describe the proposed pooling layer
defined on the space of simplicial complexes for an input
simplicial complex and vertex clustering.
We define both a topologically motivated framework for
\nervePool{} (Section \ref{sec:TopMotivtion}) and the
equivalent matrix representation of this method (described in
Section \ref{sec:MatImplement}).
In Fig.~\ref{fig:example_equivalence}, we visually depict
both the topological and matrix formulations, and their
compatible output for an example simplicial complex and
vertex clustering.
\begin{figure}
	\centering
	\includegraphics[width=\textwidth]{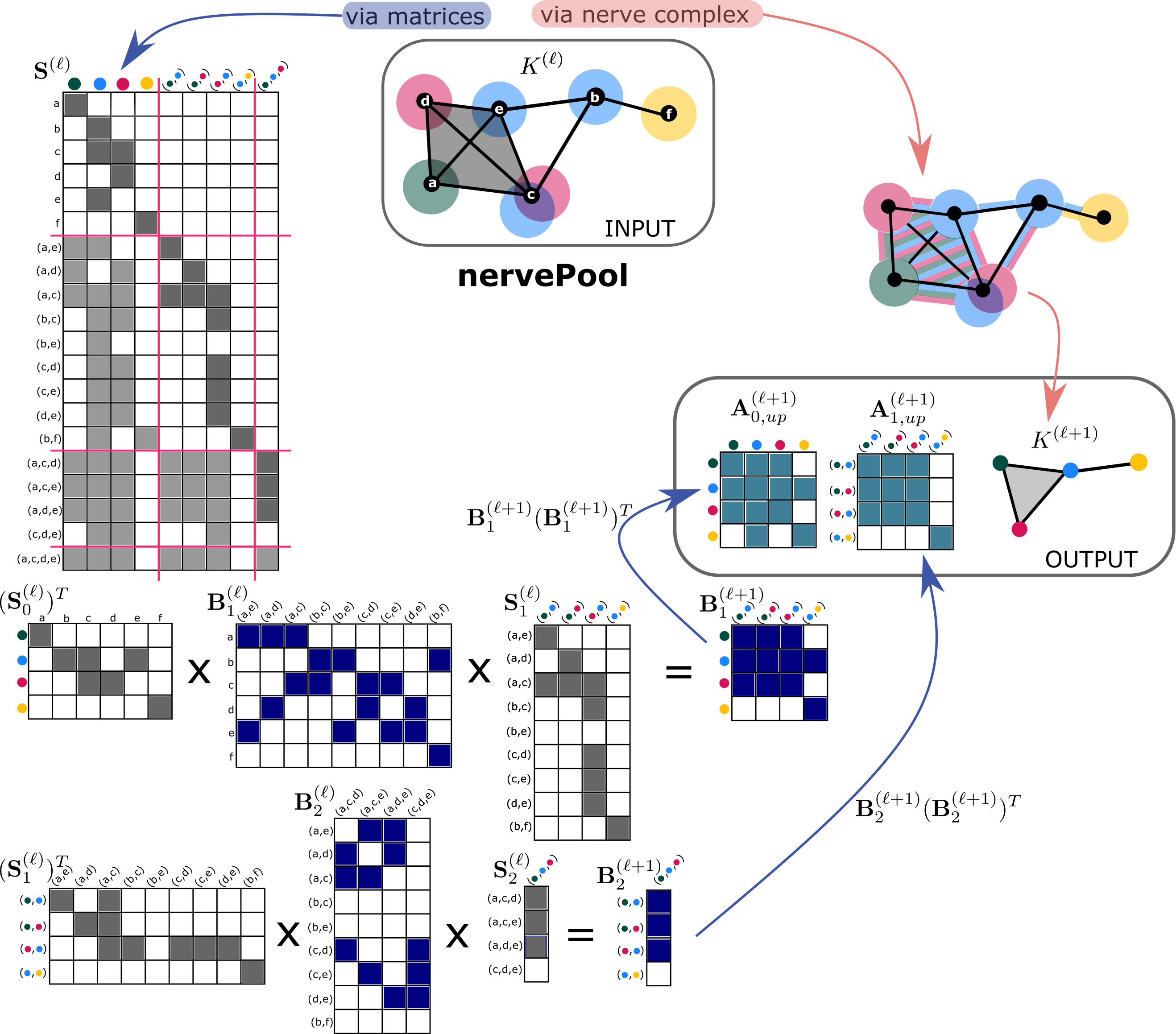}
	\caption{A visual representation of \nervePool{} for an
		example simplicial complex and choice of \textbf{soft
			partition of vertices}. Shaded-in matrix entries indicate
		non-zero values, and elsewhere are zero-valued entries.
The dark grey entries within the $\Sb^{(\ell)}$ matrix
		relate the $k$-dimensional simplices with the pooled 
    virtual simplices of the same dimension.
The output simplicial complex given by the matrix implementation is equivalent to
		the nerve complex output, up to weighting of $p$-simplices
		and the addition of ``self-loops", as indicated by non-zero
		entries on the diagonal of adjacency matrices.}
	\label{fig:example_equivalence}
\end{figure}
Notation and descriptions used for simplicial complexes are
outlined in Table \ref{Tab:notation}.
\begin{table}
	\centering
	\begin{tabular}{llp{0.65\linewidth}}
		\toprule
		{\small\sc Notation}                           & {\small\sc Description}                     \\
		\midrule
		${K}^{(\ell)}$                                 & A simplicial complex at layer $\ell$        \\
		\midrule
		$n_p^{(\ell)}$                                 & Number of $p$-dim simplices in $K^{(\ell)}$ \\
		\midrule
		$\mathcal{P}^{(\ell)}$                         & $\text{dim}\left(K^{(\ell)}\right)$         \\
		\midrule
		$\mathcal{N}^{(\ell)} =
		\sum_{p=0}^{\mathcal{P}^{(\ell)}}n_p^{(\ell)}$ & Total
		number of simplices for $K^{(\ell)}$                                                         \\
		\midrule
		$d_p^{(\ell)}$                                 & Number of features on $p$-dim simplices     \\
		\bottomrule
	\end{tabular}
	\vspace{1em}
	\caption{Simplicial complex notation and descriptions}
	\label{Tab:notation}
\end{table}

The input to \nervePool{} is a simplicial complex and a
learned partition of the vertices (i.e.~a specified cover on
the vertex set).
For exposition purposes, we assume that the initial clusters
can form a soft partition of the vertex set, meaning every
vertex is assigned to at least one cluster but vertices can
have membership in more than one cluster.
However, in Section \ref{sec:properties}, theoretical proof
of some properties require restriction to the setting of a
hard partition of the vertex set.
The initial vertex clusters give us a natural way to coarsen
the underlying graph ($1$-skeleton) of the input simplicial
complex, where clusters of vertices in the original complex
are represented by meta-vertices in the pooled complex.
However, in order to collapse both structural and attributed
information for higher-dimensional simplices, we require a
method to extend the clusters.
\nervePool{} provides a mechanism in which to naturally
extend graph pooling methods to apply on higher-dimensional simplices.

\subsection{Topological Formulation}
\label{sec:TopMotivtion}
Before defining the matrix implementation of simplicial
pooling, we will first describe the topological formulation
using the nerve complex.
\nervePool{} follows an intuitive process built on learned
vertex cluster assignments which is extended to higher
dimensional simplices in a deterministic fashion.

From the input cover on just the vertex set, however, we lack
cluster assignments for the higher dimensional simplices.
To define a coarsening scheme on the entire complex, we must
extend the cover such that each of the simplices is included
in at least one cluster using a notion of local neighborhoods
around each simplex.
A standard way to define local neighborhoods in a simplicial
complex is the star of a simplex, which can be defined on
simplices of any dimension.
For \nervePool{}, we only require the star defined on
vertices, $St(v) = \{\sigma_p\in K^{(\ell)} \mid v \subseteq
	\sigma_p \}$, which is the set of all simplices in
$K^{(\ell)}$ such that $v$ is a vertex of $\sigma_p$.
Using this construction, for each given vertex cluster
$U_i=\{v_0,\hdots,v_N\}$, we extend it to the union of the
stars of all vertices in the cluster,
\begin{equation*}
	\widetilde{U}_i = \bigcup_{v\in U_i}St(v)\, .
\end{equation*}
The resulting cover of the complex, $\mathcal{U} =
	\{\widetilde{U}_i\}$, is such that all simplices in
$K^{(\ell)}$ are part of at least (but often more than) one
cover element $\widetilde{U}_i$.
Figure \ref{fig:star_of_simplex} shows an example of a
cluster consisting of three vertices and the simplices
$\widetilde{U}_i$ which contribute to $\sigma \in K^{(\ell+1)}$.
\begin{figure}
	\centering
	\includegraphics[width=.6\linewidth]{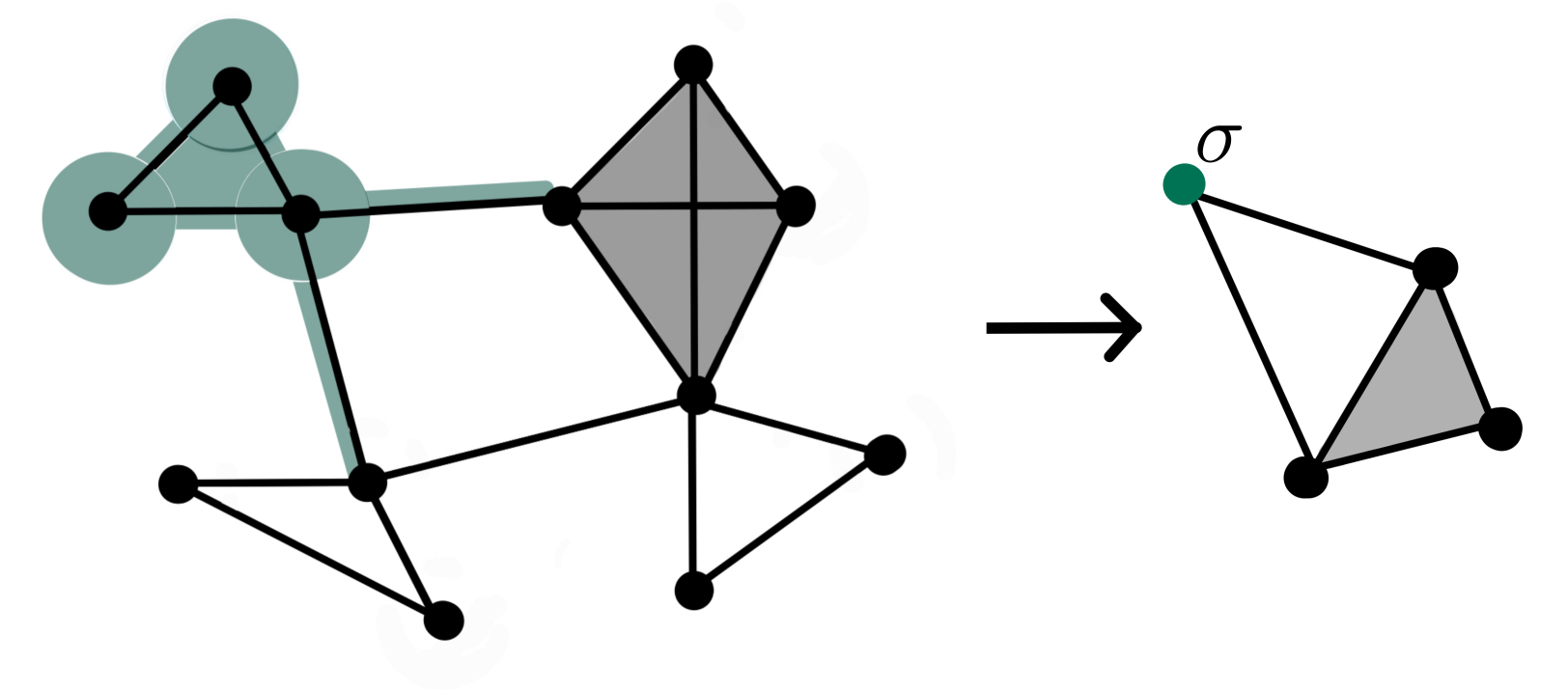}
	\caption{The three green highlighted vertices come from a
		single cover element $U_i$ of a given example cover; the
		full cover example will be continued in
		Fig.~\ref{fig:sx_pool_diagram}.
The extended cover $\widetilde{U_i}$ for this collection
		of vertices, defined by the union of the star of every
		vertex in the cluster, is shown by the highlighted green simplices.
These simplices in $K^{(\ell)}$ all contribute information
		to the meta vertex $\sigma \in K^{(\ell+1)}$.}
	\label{fig:star_of_simplex}
\end{figure}
Note that neither $St(v)$ nor $\widetilde{U}_i$ are
simplicial complexes, because they are not closed under the
face relation.

We use this cover of the complex to construct a new complex,
using the nerve. Given any cover $\mathcal{A} = \{A_i\}_{i\in
	I}$, the nerve is defined to be the simplicial complex given by
\begin{equation*}
\mathrm{Nrv}(\mathcal{A}) = \Bigl{\{} \sigma \subseteq I \;
	| \; \bigcap_{i\in\sigma} A_i \neq \emptyset \Bigl{\}} \, ,
\end{equation*}
where each distinct cover element $A_i$ is represented as a
vertex in the nerve.
If there is at least one simplex in two distinct cover
elements $A_i$ and $A_j$, we add corresponding edge
$(A_i,A_j)$ in the nerve complex.
Similarly, for higher dimensions if there exists a $p$-way
intersection of distinct cover elements, a
$(p-1)$-dimensional simplex is added to the nerve complex.
This general nerve construction gives us a tool to take the
extended cover of the simplicial complex $\mathcal{U} =
	\{\widetilde{U}_i\}$ and define a new, pooled simplicial complex:
\begin{equation*}
	K^{(\ell+1)} := \mathrm{Nrv}(\mathcal{U})\,.
\end{equation*}
Figure \ref{fig:sx_pool_diagram} shows the nerve of a
$3$-dimensional simplicial complex using four predetermined
clusters of the vertices.

\begin{figure}
	\centering
	\includegraphics[width=.8\linewidth]{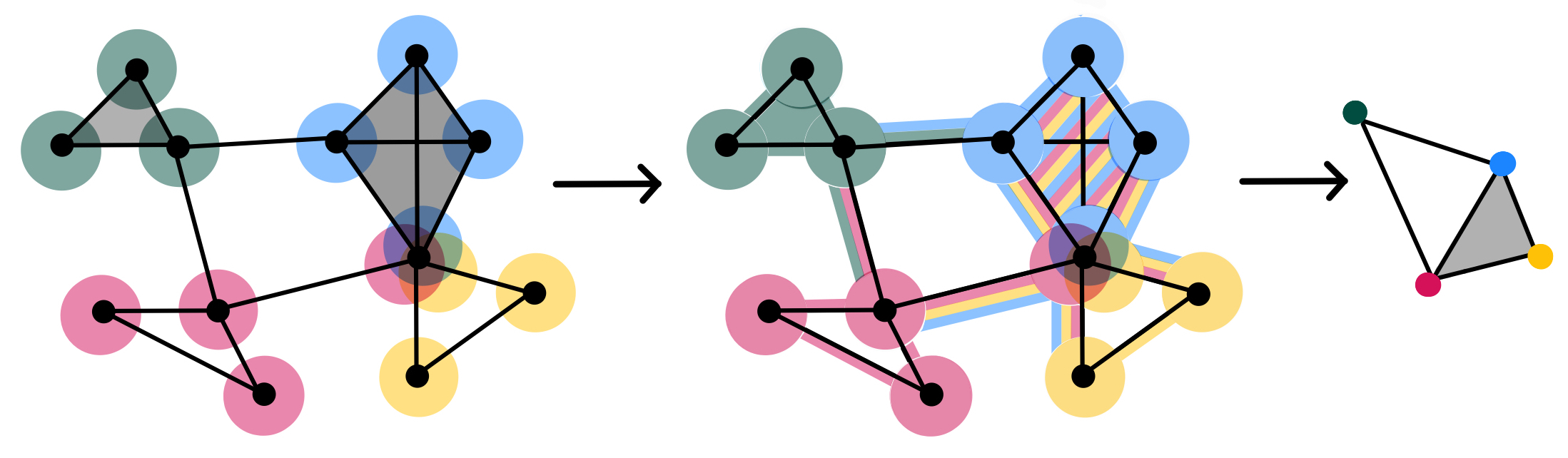}
	\caption{Illustration of \nervePool{} on an example $3$-dim
		simplicial complex, applied to coarsen the complex into a
		$2$-dimensional simplicial complex. The left-most complex
		is the input simplicial complex, with vertex cluster
		membership indicated by color. The center complex depicts
		the extended clusters $\mathcal{U} = \{\widetilde{U}_i\}$
		and the right-most complex is the pooled simplicial
		complex, determined by $\Nrv(\mathcal{U})$.}
	\label{fig:sx_pool_diagram}
\end{figure}
A notable feature of the extended covers that facilitate
\nervePool{} is that they are based on the learned vertex
cluster assignments and as such, cover elements are not
necessarily contractible.
Vertex clusters that are spatially separated on the complex
can result in cover elements that are non-convex.
Thus, since open cover elements are not guaranteed to be
convex, we cannot guarantee that the simplicial complex
$K^{(\ell+1)} = \Nrv(\mathcal{U})$ is
homotopy-equivalent to the original complex $K^{(\ell)}$ as
would be needed to apply the nerve lemma \cite{weil1952theoremes}.

\subsection{Matrix Implementation}
\label{sec:MatImplement}

In practice, we apply the simplicial pooling method described
above through matrix operations on constructed cluster
assignment matrices, boundary matrices, and simplex feature matrices.
Matrix notation used to define simplicial pooling is outlined
in Table \ref{Tab:matrices}.
See also  Table \ref{Tab:notation} for simplicial complex notation.
\begin{table}
	\centering
	{
		\renewcommand{\arraystretch}{1.2}\begin{tabular}{llp{0.6\linewidth}}
			\toprule
			{\small\sc Notation}
			                        & {\small\sc Dimension}
			                        & {\small\sc Purpose}                     \\
			\midrule

			$\Ab_p^{(\ell)}$        & $n_p^{(\ell)} \times n_p^{(\ell)}$    &
			Adjacency matrix for $p$-dim simplices at layer $\ell$            \\
\midrule
			$\Bb_p^{(\ell)}$        & $n_{p-1}^{(\ell)} \times
			n_{p}^{(\ell)}$         & $p$-dimensional boundary matrix. Maps
			features on $p$-dim simplices to the space of $(p
			-1)$-dim simplices                                                \\
			\midrule
			$\Xb_p^{(\ell)}$        & $n_p^{(\ell)} \times d_p^{(\ell)}$    &
			Features on $p$-dim simplices at layer $\ell$                     \\
\midrule
			$\Sb^{(\ell)}$          & $\mathcal{N}^{(\ell)} \times
			\mathcal{N}^{(\ell+1)}$ & Block matrix for pooling
			simplices at layer $\ell$                                         \\
			\midrule
			$\Sb_{q,p}^{(\ell)}$    & $n_q^{(\ell)} \times
			n_p^{(\ell+1)}$         & Sub-block of $\Sb^{(\ell)}$ matrix
			which maps $q$-simplices in $K^{(\ell)}$ to $p$-simplices
			in the pooled complex $K^{(\ell+1)}$. If $q=p$, we write
			$\Sb_q^{(\ell)}$                                                  \\
			\toprule
		\end{tabular}}
	\vspace{1em}
	\caption{Matrix notation with dimensions}
	\label{Tab:matrices}
\end{table}
 For this paper, we assume that the input boundary matrices
that represent each simplicial complex are non-oriented,
meaning the matrix values are all non-negative.
Since the nerve complex is an inherently non-oriented object
(built using intersections of sets of simplices) and
\nervePool{} is formulated based on a nerve construction, it
is necessary that the boundary matrices we use to represent
each simplicial complex similarly do not take orientation into account;
we use $\left|\Bb_p\right|$ for \nervePool{} pooling operations.
However, if the other layers of the neural network that the
pooling layer is used within require orientation of
simplices, it is sufficient to fix an arbitrary orientation
(so long as it is consistent between dimensions of the complex).
For example, using an orientation equivariant message passing
simplicial neural network (MPSN) layer, we can fix an
orientation of simplices after applying \nervePool{} without
affecting full network invariance with respect to orientation
transformations \cite{bodnar2021simplicial}.
\paragraph*{Extend vertex cluster assignments}
From the input vertex cover, we represent cluster membership
of each $0$-simplex (vertex) in an assignment matrix $\Sb_0^{(\ell)}$.
This $n_0^{(\ell)} \times n_0^{(\ell+1)}$ matrix keeps track
of the vertex clusters with entries
$\Sb_0^{(\ell)}[\sigma_0^{(\ell)},U_i] >0$ if the vertex
$\sigma_0^{(\ell)}$ is in cluster $U_i$ and zero entries elsewhere.
Each row of $\Sb_0^{(\ell)}$ corresponds to a vertex in the
original complex, and each column corresponds to a cover
element (or cluster), represented as a meta-vertex in the
next layer $\ell+1$.
Using these initial clusters which consist of vertices only,
we then extend the cluster assignments to include all of the
higher dimensional simplices of the complex in a deterministic way.
From the original matrix
$\Sb_0^{(\ell)} :n_0^{(\ell)}\times n_0^{(\ell+1)}$,
we extend the pooling to all simplices resulting in the full
matrix, $\Sb^{(\ell)}:\mathcal{N}^{(\ell)}\times
	\mathcal{N}^{(\ell+1)},$
where each of its sub-blocks
$\Sb_{q,p}^{(\ell)}:n_q^{(\ell)}\times n_p^{(\ell+1)}$
map $q$-simplices in $K^{(\ell)}$ to $p$-simplices in the
coarsened complex $K^{(\ell+1)}$.
Figure \ref{fig:Smatrix} shows a visualization of this matrix
extension from vertex cluster assignments to all dimension
simplex cluster assignments.
\begin{figure}
	\centering
	\includegraphics[width=.5\textwidth]{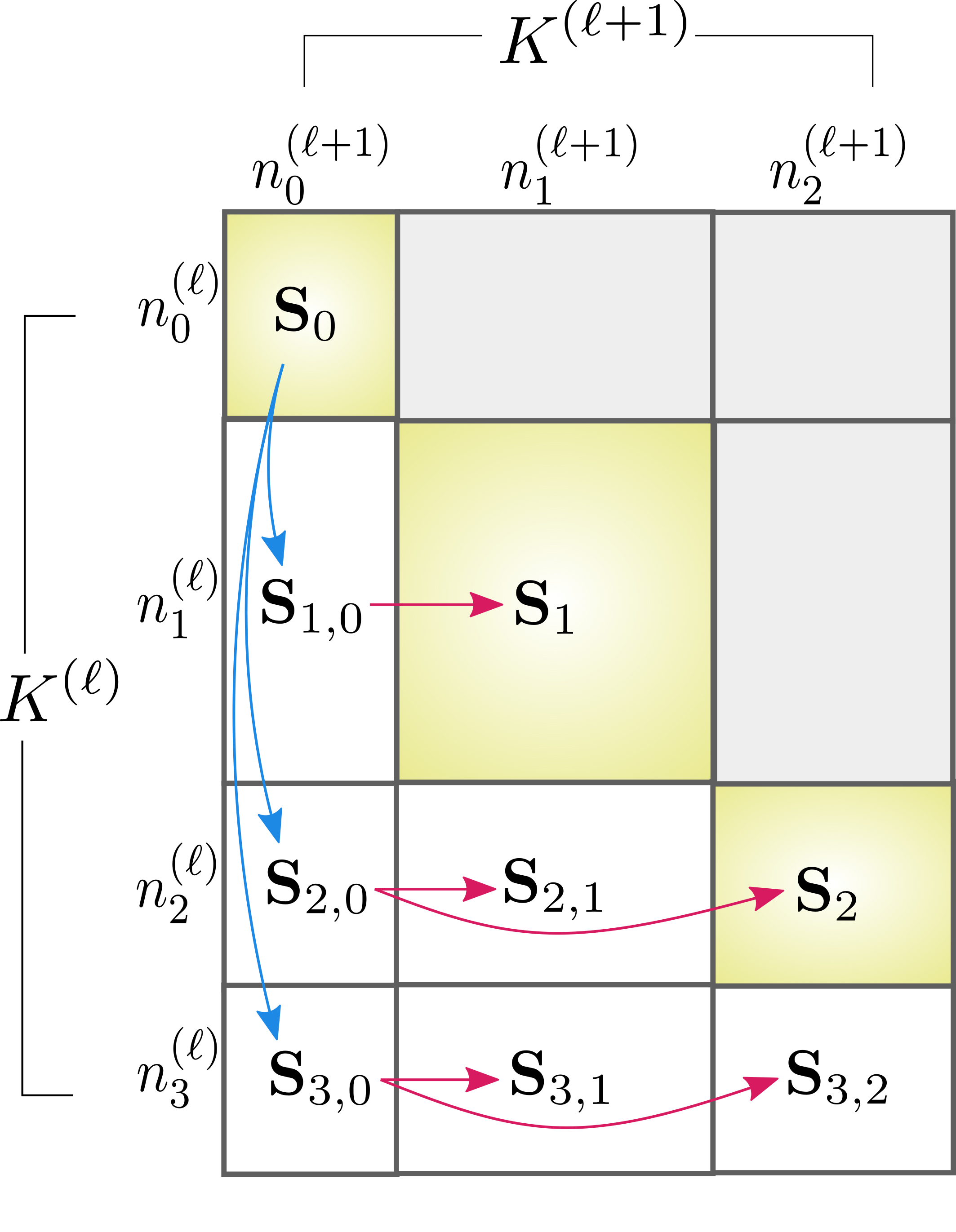}
	\caption{Visual depiction of $\Sb^{(\ell)}:
			\mathcal{N}^{(\ell)} \times \mathcal{N}^{(\ell+1)}$ block
		matrix. Sets of sub-blocks are used to map simplices of
		original simplicial complex to the pooled simplices.
}
	\label{fig:Smatrix}
\end{figure}
As visualized in this diagram, there are two directions by
which information is cascaded through the matrix
$\Sb^{(\ell)}$: extend \textit{down} and extend
\textit{right}, both of which pass information to simplices
of the next higher dimension.
We pass pooling information to simplices of the next higher
dimension, within the original simplicial complex via the
\textit{down arrow update},  $\Sb_{0}
	\textcolor{CBF-blue}{\rightsquigarrow} \,\Sb_{q,0}, $ for $
	q>0$ defined,
\begin{equation}
	\label{eqn:downfunction}
	\Sb_{q,0}[\sigma_q,v_j] :=
	\begin{cases}
		1 & \text{if } \Sb_0[v_m,v_j] = 1 \text{ for some } v_m\in
		\sigma_q                                                  \\
		0 & \text{ otherwise.}
	\end{cases}
\end{equation}
Pooling information is passed to simplices of higher
dimensions for the new simplicial complex via the
\textit{right arrow update}, $\Sb_{q,0}
	\textcolor{CBF-pink}{\rightsquigarrow}\, \Sb_{q,p}$, for $0<p\leq q$
\begin{equation} \label{eqn:rightfunction}
	\Sb_{q,p}[\cdot ,\sigma_{p}] := \Sb_{q,0}[\cdot, v_a] \odot
	\Sb_{q,0}[\cdot, v_b]\odot \cdots \odot \Sb_{q,0}[\cdot, v_c] \, ,
\end{equation}
for $\sigma_{p}$ the simplex given by $(v_a,v_b,\hdots,v_c)$
and where $\odot$ indicates the pointwise multiplication of
matrix columns.
If $\Sb_{q,p}[\cdot, \sigma_{p}]= \bm{0}$ (the zero vector)
then $\sigma_{p} \notin K^{(\ell+1)}$ and that column is not
included in $\Sb_{q,p}$.

\paragraph*{Pool with full cluster assignment matrix}
\revision{
Sub-blocks of the pooling matrix correspond to mapping
$p$-dimensional simplices in $K^{(\ell)}$ to pooled
$p$-simplices in $K^{(\ell+1)}$.
In particular, the $p$-dimensional boundary matrices of the
pooled simplicial complex are computed using
\begin{equation*}
	\Bb_p^{(\ell +1)} = \left(\Sb_{p-1}^{(\ell)}\right)^T
	\Bb_p^{(\ell)} \Sb_{p}^{(\ell)}.
\end{equation*}
In principle, using the full cluster assignment matrix allows 
the interaction of the lower dimensionsal simplices with the higher 
order ones and vice versa. 
By setting the sub- and super diagonal blocks to zero, we can thus 
restrict \nervePool{} to only allow interaction between the simplices 
of the same dimension.
In this work, we show that this restriction is sufficient to preserve the topological structure of the simplicial complex.
However, since the full cluster assignment matrix is available, it is possible to explore the use of the full matrix which we utilize in experiments in Section \ref{sec:Experiments}. 
}

The pooled simplicial complex is entirely determined by
the upper adjacency matrices (Equation \ref{eqn:Aup}),
$\Ab_{up,p}^{(\ell +1)} = \Bb_{p+1}^{(\ell
		+1)}\left(\Bb_{p+1}^{(\ell +1)}\right)^T$, either on their
own or normalized using the absolute difference with the
degree matrix $\Db_p = \sum_{j=0}^{n_p}B_{p+1}$,
\begin{equation*}
	\Ab_{up,p}^{(\ell +1)} = \left|\, \Db_p - \Bb_{p+1}^{(\ell
		+1)}\left(\Bb_{p+1}^{(\ell +1)}\right)^T \right| \,.
\end{equation*}
The set of pooled boundary matrices, or alternatively the set
of upper adjacency matrices, give all of the structural
information necessary to define a simplicial complex.
Note that the specific adjacency representation is a choice,
and one could use alternative notions of adjacency described
in Section \ref{sec:SimplicialComplexes} to summarize the
pooled simplicial complex.
Simplex embeddings are also aggregated according to the
cluster assignments in $\Sb^{(\ell)}$ to generate coarsened
feature matrices for simplices in layer $\ell +1$.
We compute these pooled embeddings for features on
$p$-simplices by setting
\begin{equation}\label{eqn:poolingX}
	\Xb_p^{(\ell +1)} = (\Sb_p^{(\ell)})^T \Xb_p^{(\ell)} \, ,
\end{equation}
where $\Xb_p^{(\ell)}$ are the features on $p$-simplices.
The coarsened boundary matrices and coarsened feature
matrices for each dimension can then be used as input to
subsequent layers in the network.

The matrix implementation of \nervePool{}, using cluster
assignment matrices multiplied against boundary matrices, is
consistent with the set-theoretical topological nerve
construction outlined in Section \ref{sec:TopMotivtion}.
Given the same complex and initial vertex cover, the output
pooled complex using the nerve construction is structurally
equivalent to the simplicial complex described by
$\Ab_{up,p}^{(\ell +1)}$.
In Theorem \ref{thm:topEquivMat} and associated proof, we will give a formal description of this equivalence.

\section{\nervePool{} Properties}
\label{sec:properties}

In this section, we describe observational and theoretical
properties of \nervePool{} as a learned simplicial complex
coarsening method, including a formally defined identity function.
In practice, most learned graph pooling methods output soft
partitions of the vertex set, meaning a vertex can be
included in multiple clusters.
However, extraneous non-zero entries in the ``boundary''
matrices\footnote{We use quotes around boundary because the
	matrices generated by \nervepool{} are not necessarily true boundary
	matrices. They are boundary matrices with possible extra
	non-zero entries which do not affect the output adjacency
	representations.} built by \nervepool{} emerge from vertices
assigned to multiple clusters.
While the result is a matrix that does not in fact represent
a simplicial complex, they do not seem to pose an issue in
practice, so long as the output simplicial complex is
represented by its adjacency matrices, and not the ``boundary'' matrices.
The unintended entries of ``boundary'' matrices do not affect
the adjacency matrix representation because if one is present
in the matrix, then there must also be non-zero entries
present for the edges that appear in the adjacency matrix.

As a result, in practice \nervePool{} can be used in a hard-
or soft-partition setting.
However to prove equivalence of the nerve construction to the
matrix implementation and simplex-permutation invariance, we
must assume a hard vertex clustering, i.e.~where each vertex
is a member of exactly one cluster.
In this regime, \nervePool{} has more theoretical
justification and we can derive the matrix implementation in
a way such that it produces the same pooled simplicial
complex output as the nerve construction, provided they are
given the same input complex and vertex cover.
Additionally, assuming a hard partition of the vertex set
allows us to prove it is a simplex permutation invariant pooling layer.
To work in this more restricted setting, given a
soft-partition input, one could simply threshold
appropriately to ensure that each vertex is assigned to a
single cluster, choosing the cluster it is included in with
the highest probability.
See the example of
Fig.~\ref{fig:example_equivalence_hardpartition} for the
resulting matrices in the case of a hard partition, which
uses the same complex as the soft partition example of
Fig.~\ref{fig:example_equivalence}.
In this restricted case, the boundary matrix output of
\nervePool{} does not have any extra non-zero entries that we
see in the case of soft partitions on the vertex set.

\begin{figure}
	\centering
	\includegraphics[width=\textwidth]{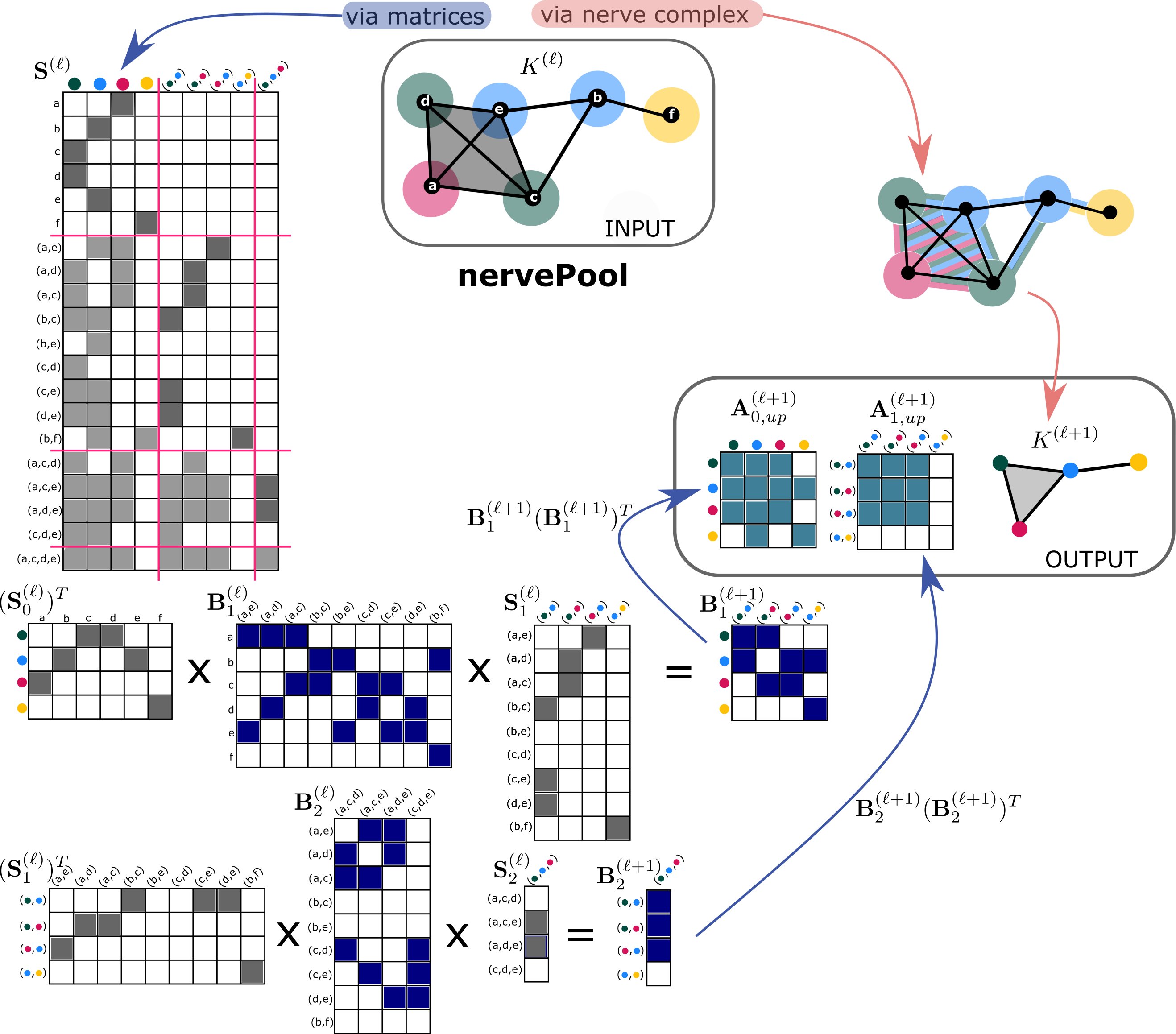}
	\caption{A visual representation of \nervePool{} for an
		example simplicial complex and choice of vertex clusters.
Note this is the same simplicial complex as
		Fig.~\ref{fig:example_equivalence}, but using a
		\textbf{hard partition of the vertex set}. Matrix entries
		indicate non-zero values, and elsewhere are zero-valued
		entries. 
}
	\label{fig:example_equivalence_hardpartition}
\end{figure}

We note one technical distinction between the standard graph
representation as an adjacency matrix, and the generalized
adjacency matrices we use here.
In the standard literature, the diagonal of the adjacency
matrix of a graph will be zero unless self loops are allowed.
However, the adjacency matrices constructed here have a
slightly different viewpoint.
In particular (at least in the hard partition setting), an
entry of an upper-adjacency matrix for dimension $p$
simplices is non-zero if the two simplices share a common
higher dimensional coface.
In the case of a graph, the $0$-dimensional upper-adjacency
matrix would then have non-zero entries on the diagonal since
a vertex that is adjacent to any edge is thus upper adjacent to itself.

Another important technical note is that with no assumptions
on the input clustering, we cannot make promises about
maintaining the topological structure of the simplicial
complex heading into the pooled layer.
Since the vertex cluster assignments are largely
task-dependent, any collection of vertices can be grouped
together in a cluster if it minimizes the loss function with
respect to a given task, with no regard for their spatial
locality on the simplicial complex.
This can lead to potentially un-intuitive behaviour in the
simplex pooling for those used to working in the setting of
the nerve lemma \cite{weil1952theoremes}, specifically
learned vertex groupings which are not localized clusters.
For example, Fig.~\ref{fig:homology_diff_fig} shows three
different initial vertex covers of same input simplicial
complex, resulting in three different \nervePool{} outcomes
of simplicial complexes with different homology.
\begin{figure}
	\centering
	\includegraphics[width=.5\textwidth]{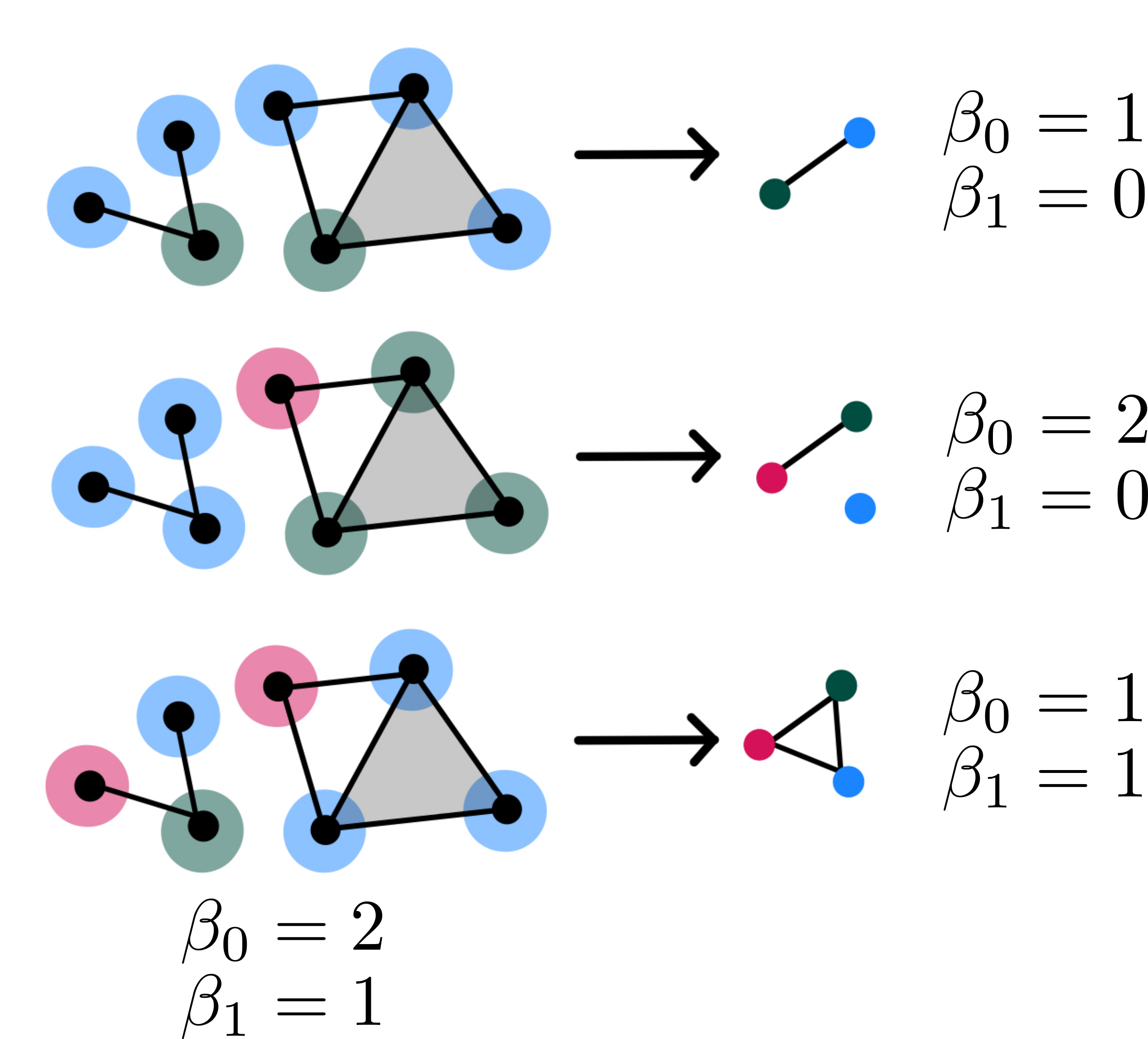}
	\caption{An input simplicial complex with three different
		choices of initial vertex covers (top, middle, bottom).
Each of these cover choices produce pooled simplicial
		complexes of different homology using \nervePool{}, as
		indicated by the Betti numbers for dimensions $0$ and $1$.
The first initial vertex cover (top) produces \nervePool{}
		output which changes Betti numbers in both dimensions. The
		second cover (middle) produces an output complex with the
		same $0$-dim Betti number. The third cover (bottom)
		produces an output complex with the same $1$-dim Betti number.}
	\label{fig:homology_diff_fig}
\end{figure}

In addition to interpretability of the diagonal entries, the
hard partition assumption gives further credence for
connecting the topological nerve/cover viewpoint with the
matrix implementation presented.
We prove this connection in the next theorem, which hinges on
the following construction.
Assume we are given the input cover of the vertex set
$\{U_i\}_{i=1}^{n_0^{(\ell+1)}}$, have constructed the cover
of the simplicial complex using the stars of the vertices
$\cU = \{\widetilde{U}_i \}_{i=1}^{n_0^{(\ell+1)}}$ and have
built the pooled simplicial complex using the nerve
construction $K^{(\ell+1)} = \Nrv(\cU)$.
There is a natural simplicial map from the simplicial complex at layer $\ell$ to that at layer $\ell+1$, $f:K^{(\ell)} \to K^{(\ell+1)}$
where a vertex $v \in K^{(\ell)}$ is mapped to
the vertex in $K^{(\ell+1)}$ representing the (unique due to the hard partition assumption) cover element containing it.
As with any simplicial map, simplices in $K^{(\ell)}$ may be
mapped to simplices in $K^{(\ell+1)}$ of a strictly lower dimension.
\revision{
In the theoretical results, we make the  assumption of using only the block diagonal of the
$\Sb^{(\ell)}$ matrix, which means that the matrix construction
version will essentially ignore these lower dimensional maps,
instead focusing on the portions of the simplicial map that
maintain dimension.
That being said, one could make some different choices in the
pipeline to maintain these simplices.
The trade-off is that 
additional computation time is required, but now message passing can flow across cells of different dimensions potentially learning to better learning outcomes.
}

\begin{theorem}[Equivalence of Topological and Matrix
		Formulations]\label{thm:topEquivMat}
	Given the same input simplicial complex and hard partition
	of the vertex set (cluster assignments), the topological
	nerve/cover viewpoint and matrix implementation of
	\nervePool{} produce the same pooled simplicial complex.
\end{theorem}

\begin{proof}
	For the purposes of this proof, we assume that a hard partition was provided with all vertices given equal weight.
	This means, that we assume it was given in the form of a binary matrix $\Sb^{(\ell)}_0$ where $\Sb^{(\ell)}[v,j] = 1$ iff vertex $v \in K^{(\ell)}$ is in cluster $U_j$, and all rows have exactly one 1.
	A similar proof applies where we assume exacty one entry $\Sb^{(\ell)}[v,j] >0$ in any column, with the technicality that checking wether matrices are equal is replaced with checking whether matrices have non-zero entries (although potentially not equal) in the same locations.
	Since our theorem corresponds to the complex constructed rather than to the aggregation of the features, the binary assumption is enough for our purposes.

	The proof that these two methods result in the same simplicial complex is broken into two parts.
	\begin{romanenumerate}
		\item \textbf{The maps are the same.}
		We will show that the matrix representation $\mathbf{M}$
		of the natural simplicial map $f:K^{(\ell)}\rightarrow
			\Nrv(\cU) = K^{(\ell+1)}$ discussed above, is equal to
		the matrix $\Sb^{(\ell)}$ of simplex cluster assignments
		used for pooling in the matrix formulation.

		\item \textbf{The complexes (boundary matrices) have the same locations of non-zero entries.}
		Since the complexes are built based on the locations of non-zero entries in $\Sb$, we show that although they might not be equal, the boundary matrices
which represent the nerve of the cover $\cU$
		have non zero entries in the same locations as
		the boundary matrices
computed using the boundary and cluster assignment
		matrices from layer $\ell$.
That is, we will show
		\begin{equation*}
			\Bb_{\Nrv(\cU),p}[\tau_{p-1}, \tau_p] \neq 0 \qquad \text{iff} \qquad  \left(\Sb_{p-1}^{(\ell)}\right)^T
			\Bb_p^{(\ell)} \Sb_{p}^{(\ell)}[\tau_{p-1}, \tau_p] \neq 0
		\end{equation*}
		and thus represent the same complex which we can denote
		as $K^{(\ell+1)}$.
\end{romanenumerate}

	\textit{Proof of (i):}
	We prove equivalence of the simplicial map
	$f:K^{(\ell)}\rightarrow K^{(\ell+1)}$ to the matrix $\Sb^{(\ell)}$ of simplex cluster assignments by showing equivalent sparsity patterns of their matrices.
	The simplicial map is defined as
	$f(\sigma) = \{f(v)\}_{v\in\sigma}$, removing duplicates where they occur and viewing the result as a simplex in $K^{(\ell+1)}$.
	Note that $f(v)$ is some vertex $U_i$ in $K^{(\ell+1)}$ where $v\in U_i$.
	Let $\mathbf{M}$ be a matrix that represents this
	operator, defined by
	$\mathbf{M}[\sigma, \tau] = 1$
	iff $\tau \leq f(\sigma)$, i.e.~if $\tau$ is a face of $f(\sigma)$.

	Denote $\sigma = \{v_0,\cdots,v_q\}$ and $\tau = \{w_0,\cdots,w_p\}$, and then denote the clusters associated to the $w$'s as $\{U_0,\cdots,U_p\}$ respectively.
	Note that because $\tau \in K^{(\ell+1)}$, we know that $\bigcap_{j=0}^p \tilde U_j \neq \emptyset$.
	We now check that
	$\mathbf{M}[\sigma,\tau] = \Sb[\sigma,\tau]$
	for all pairs of $\sigma \in K^{(\ell)}$ and $\tau \in K^{(\ell+1)}$.
	By definition, $\mathbf{M}[\sigma, \tau] = 1$ iff for every $w_j \in \tau$ with associated cover $U_j$, there is a $v_i \in \sigma$ with $v_i \in U_j$.
	Then by Eqn.~\eqref{eqn:downfunction}, this occurs iff for any $w \in \tau$,
	$\Sb_{q,0}[\sigma, w] = 1$.
	Following Eqn.~\eqref{eqn:rightfunction} and because this is true for all $w \in \tau$, we have $\Sb_{q,p}[\sigma,\tau] = 1$.

	\textit{Proof of (ii):}
	Assume we are given a pair $\tau_{p-1} \leq \tau_p$,
	simplices in $\Nrv(\cU)$ of dimension $p-1$ and $p$ respectively, and thus
	\begin{equation*}
		\Bb_{\Nrv(\cU),p}[\tau_{p-1}, \tau_p]= 1.
	\end{equation*}
	We will show that if this occurs, the matrix
	multiplication results in a nonzero entry
	$\Bb_{p}^{(\ell+1)}[\tau_{p-1}, \tau_p]$ as well.
	After potentially reindexing, we can write $\tau_p = \{w_0,\cdots,w_p\}$ and $\tau_{p-1} = \{w_0,\cdots,w_{p-1}\}$.
	Similar to before, we denote the associated cover elements as $\{U_0, \cdots, U_p\}$.
	Since $\tau_p$ is in the nerve, we know that $\bigcap_{j=0}^p \tilde U_p \neq \emptyset$.
	Combining this with the definition of the star, this implies there is a simplex $\sigma_p = \{v_0,\cdots,v_p\} \in K^{(\ell)}$ where $v_j \in U_j$ for all $j$.
	Further, defining $\sigma_{p-1} = \{v_0,\cdots,v_{p-1}\}$, we have that $\sigma_{p-1} \in \bigcap_{j =0}^{p-1} U_j$.

	By construction $\sigma_{p-1}\leq \sigma_p$, so
	$\Bb_{p}^{(\ell)}[\sigma_{p-1}, \sigma_{p}] = 1$.
	Next, we check
	$\Sb_{p}^{(\ell)}[\sigma_p, \tau_{p}]$.
	In this case, because for every $w_i \in \tau_p$ there is $v_i \in \sigma_p$ with $v_i \in U_i$, we have
	$\Sb_p^{(\ell)}[\sigma_p, w_i] = 1$ for all $w_i$.
	And as this is true for all $w_i \in \tau_p$, it implies that $\Sb_{p}^{(\ell)} [\sigma_p, \tau_p] = 1$.
	Similar logic means that $\Sb_{p-1,p-1}^{(\ell)} [\sigma_{p-1}, \tau_{p-1}] = 1$.
	The resulting matrix multiplication
	\begin{equation*}
		\Bb_p^{(\ell +1)} = \left(\Sb_{p-1}^{(\ell)}\right)^T
		\Bb_p^{(\ell)} \Sb_{p}^{(\ell)}
	\end{equation*}
	will thus also have a non-zero entry in $[\tau_{p-1},\tau_p]$.

	Conversely, assume $\tau_{p-1}$ is not a face of
	$\tau_p$ and so
	\begin{equation*}
		\Bb_{\Nrv(\cU),p}[\tau_{p-1}, \tau_p]= 0.
	\end{equation*}
	We need to show $\Bb_p^{(\ell+1)}[\tau_{p-1},\tau_p]= 0$.
	Unfolding the matrix multiplication means this requires
	\begin{equation*}
		\sum_{\sigma_{p-1} \in K^{(\ell)}} \sum_{\sigma_p \in K^{(\ell)}}
		\left(\Sb_{p-1}^{(\ell)}\right)^T [\tau_{p-1}, \sigma_{p-1} ]
		\cdot
		\Bb_p^{(\ell)} [\sigma_{p-1}, \sigma_p]
		\cdot
		\Sb_{p}^{(\ell)}[ \sigma_p, \tau_p ]
	\end{equation*}
	to be 0, and thus at least one of the three entries must be zero for all pairs of $\sigma_{p-1}$ and $\sigma_p$.

	Fix $\sigma_{p-1}$ and $\sigma_p$, and note that the middle term is 0 if $\sigma_{p-1}$ is not a face of $\sigma_p$, so we can assume $\sigma_{p-1} \leq \sigma_p$.
	If  $\Sb_{p}^{(\ell)}[\sigma_p, \tau_{p}] =0$ we are done,
	so assume $\Sb_{p}^{(\ell)}[\sigma_p, \tau_{p}] \neq 0$.
By part (i), this means that $\tau_p \leq f(\sigma_p)$ and so by the dimensionality assumptions, we must have
	$\tau_p  = f(\sigma_p)$.
	Then each vertex of $ \sigma_p$ is in a distinct cover element represented by $\tau_p$.

	As $\tau_{p-1}$ is not a face of $\tau_p$, there is at least one vertex $w$ in $\tau_{p-1}$ which is not in $\tau_p$, and denote its cover element by $U$.
	Since every vertex of $\sigma_p$ corresponds to a distinct cover element which is in $\tau_p$, this implies that the vertices of $\sigma_p$, and thus those of $\sigma_{p-1}$ cannot be in $U$.
	This means that $\tau_{p-1}$ is not a face of $f(\sigma_{p-1})$.
	Using the interpretation from part (i), this means
	$\Sb_{p-1}^{(\ell)T}[\tau_{p-1}, \sigma_{p-1}] = 0$
	as required.
\end{proof}

\begin{theorem}[\nervePool{} identity function]
	There exists a choice of cover (cluster assignments) on the
	vertices, $\Sb_0^{(\ell)}$, such that \nervePool{} is the
	identity function, i.e.~$K^{(\ell+1)} = K^{(\ell)}$ and the
	pooled complex is equivalent to the original, up to
	re-weighting of simplices. In particular, this choice of
	$\Sb_0^{(\ell)}$ mapping is such that each distinct vertex
	in $K^{(\ell)}$ is assigned to a distinct vertex cluster in
	$K^{(\ell+1)}$.
\end{theorem}
\begin{proof}
	For the proof, we show a bijection between the simplicial
	complexes using the set-theoretic representations, thus
	showing that the resulting layer is the same up to
	reweighting of the simplicies.
Suppose each vertex in the original complex is assigned to
	a distinct cluster, so that $U_i=\{v_i\}$, and then
	$\widetilde{U_{i}} = \St(v_i)$.
For any simplex $\sigma = \{v_{i_0},\cdots,v_{i_d}\}$ in
	$K^{(\ell)}$, we denote
	$\widetilde \sigma = \{\widetilde{U_{i_0}}, \cdots,
		\widetilde{U_{i_d}} \}$.
Note that $\widetilde \sigma$ must be a simplex in
	$K^{(\ell+1)}$ since $\sigma$ is in the intersection of the
	stars $\bigcap_{i = i_0}^{i_d} \St(v_i)$.

	Injectivity is immediate by definition.
We show the bijectivity of the set map $K^{(\ell)} \to
		K^{(\ell+1)}$ sending $\sigma \mapsto \widetilde \sigma$.
To check surjectivity, note that for any $\widetilde \sigma
		\in K^{(\ell + 1)}$ with $\{\widetilde{U_{i_0}}, \cdots,
		\widetilde{U_{i_d}} \}$, the intersection $\bigcap_{i =
			i_0}^{i_d} \St(v_i)$ must contain some simplex $\tau$.
By definition, this means that $v_i \leq \tau$ for all $i
		\in \{i_0,\cdots i_d\}$.
Then $\sigma = \{v_{i_0}, \cdots, v_{i_d} \} \leq \tau$,
	and by the closure of the simplicial complex, $\sigma \in K^{(\ell)}$.
\end{proof}

\subsection{Permutation invariance}
A key property of graph structured data and graph neural
networks is that the re-ordering of vertices should not
affect the network output.
Different representations of graphs by their adjacency
information are equivalent, up to arbitrary reordering of the vertices.
This property implies a permutation invariant network is
necessary for most tasks such as graph classification, and
subsequently permutation invariant pooling layers are necessary.
Other types of invariance, including orientation of
simplices, are not relevant for \nervePool{}, due to the
assumption of non-oriented simplices.

\begin{theorem}[\nervePool{} Permutation
		Invariance]\label{thm:permutationInv}
	\nervePool{} is a \textbf{permutation invariant} function
	on the input simplicial complex and partition of vertices.
Equivalently,
	\begin{equation*}
		\text{\nervePool{}}(P\Sb_0^{(\ell)}, K^{(\ell)}) =
		\text{\nervePool{}}(\Sb_0^{(\ell)}, K^{(\ell)})\,,
	\end{equation*}
	for any permutation matrix on vertices $P\in
		\{0,1\}^{n_0^{(\ell)}\times n_0^{(\ell)}}$.
\end{theorem}
\begin{proof}
	Permutations of simplices correspond to re-ordering of the
	vertex list of a complex. Note that the re-ordering of
	vertices of a simplicial complex induces a corresponding
	re-ordering of all its higher dimensional simplices.
Consider a simplicial complex defined by its set of
	simplicies $K:=\{\sigma_0,\hdots,\sigma_n\}$.
Consider an alternate labeling of all simplices in $K$ and
	call this new simplicial complex $K':=
		\{\sigma_0',\hdots,\sigma_n'\}$.
Then, there exists an isomorphism between the sets $K$ and
	$K'$, and specifically an isomorphism between the vertex sets of each.
This map between vertices in $K$ and vertices in $K'$ can
	be represented by a permutation matrix $P\in
		\{0,1\}^{n_0^{(\ell)}\times n_0^{(\ell)}}$.
Then, $ \text{\nervePool{}}(P\Sb_0^{(\ell)}, K^{(\ell)}) =
		\text{\nervePool{}}(\Sb_0^{(\ell)}, K^{(\ell)})$ since the
	simplex sets define the same simplicial complex structure
	and the vertex cluster assignments are permuted accordingly.
\end{proof}

\section{Training GNNs with \nervePool{}}
\label{sec:training}
The required input to \nervePool{} is a partition of the
vertices of a simplicial complex, learned using only the
$1$-skeleton (underlying graph) of the complex.
Since the learned aspect of \nervePool{} is limited to the
vertex cluster assignments ($\Sb_0^{(\ell)}$) and
higher-dimensional cluster assignments are a deterministic
function of $\Sb_0^{(\ell)}$,
the loss functions used for cluster assignments are defined
on only vertices and edges of the complex.
Any trainable graph pooling method that learns a soft
partition of the graph (DiffPool \cite{ying2018diffpool},
MinCutPool \cite{bianchi2020spectral}, StructPool
\cite{Yuan2020StructPool:}, etc) can be used in conjunction
with \nervePool{}.
However, there are minor adjustments necessary to modify
these graph pooling methods for application to \nervePool{}
simplicial complex coarsening (primarily in terms of
computing embeddings for multiple dimensions).

The broad categorization of graph pooling methods into
selection, reduction, and connection (SRC) computations
\cite{2021_SRCgraphpooling} provides a useful framework to
classify steps of \nervePool{} and necessary extensions for
the higher dimensional simplices.
The selection step refers to the method for grouping vertices
of the input graph into clusters, which can be used directly
in the case of \nervePool{} on the $1$-skeleton for the
initial vertex covers.
The reduction computation, which aggregates vertices into
meta-vertices (and correspondingly features on vertices
aggregated), must be generalized and computed for every
dimension of the simplicial complex.
In the case of \nervePool{}, this is done using the diagonal
sub-blocks of $\Sb^{(\ell)}$ to perform the reduce step at
each dimension.
Similarly, the connect step requires generalization to higher
dimensions and computation of new boundary maps to output the
pooled simplicial complex.
Table \ref{Tab:SRC_table} summarizes trainable graph pooling
methods, modified for use within \nervePool{} simplicial
complex pooling as the learned-vertex clustering input method.
\begin{table}
	\centering
	\resizebox{\columnwidth}{!}{
		\renewcommand{\arraystretch}{1.2}\begin{tabular}{l@{\hskip .5em}l@{\hskip .5em}l@{\hskip .5em}l}
			\toprule
			{\small\sc Method}
			                                                    & {\small\sc Select}
			                                                    & {\small\sc Reduce\,($\forall\,p$)}
			                                                    & {\small\sc Connect\,($\forall\,p$)}                   \\
			\midrule
			MinCut\,\cite{bianchi2020spectral}                  &
			$\Sb_0^{(\ell)}= \text{MLP}(\Xb_0^{(\ell)})$        &
			$\Xb_p^{(\ell+1)}=(\Sb_p^{(\ell)})^T\Xb_p^{(\ell)}$ &
			$B_p^{(\ell+1)} = (\Sb_{p-1}^{(\ell)})^T \Bb_p \Sb_p$                                                       \\
			\midrule
			NFM\,\cite{Bacciu2019_NMFgraphPool}                 & $\Ab_0^{(\ell)} = \mathbf{W}\mathbf{H} \rightarrow
			\Sb_0^{(\ell)}= \mathbf{H}^T$                       &
			$\Xb_p^{(\ell+1)}=(\Sb_p^{(\ell)})^T\Xb_p^{(\ell)}$ &
			$B_p^{(\ell+1)} = (\Sb_{p-1}^{(\ell)})^T \Bb_p \Sb_p$                                                       \\
			\midrule
			LaPool\,\cite{Noutahi2019_LaPool}                   & $
				\begin{cases}
					\mathbf{V}     & = \lVert \mathbf{L_0} \mathbf{X}_0 \rVert_d                                                \\
					\mathbf{i}     & = \{i | \forall j \in \mathcal{N}(i) :
					\mathbf{V}_i \geq \mathbf{V}_j\}                                                                            \\
					\Sb_0^{(\ell)} & =\text{SparseMax}\left(\beta\frac{\mathbf{x}_0\mathbf{x}_{0,i}^T}{\lVert\mathbf{x}_0\rVert
						\lVert\mathbf{x}_0,i\rVert}\right)
				\end{cases}$
			                                                    & $\Xb_p^{(\ell+1)}=(\Sb_p^{(\ell)})^T\Xb_p^{(\ell)}$ &
			$B_p^{(\ell+1)} = (\Sb_{p-1}^{(\ell)})^T \Bb_p \Sb_p$                                                       \\
			\bottomrule
		\end{tabular}}\vspace{1em}
	\caption{Trainable Graph Pooling Methods in the SRC framework, with necessary adjustments
		for higher-dimensional use within \nervePool{}.
MLP is a multi-layer perceptron defined on the vertex features,
		$\beta$ a regularization vector, and $\mathbf{i}$ a vector
		of indices.
Table adapted from \cite{2021_SRCgraphpooling}.
}\label{Tab:SRC_table}
\end{table}
 
\subsection{\nervePool{} extended using DiffPool graph pooling}
For DiffPool \cite{ying2018diffpool} graph pooling, the
vertex clusters are learned using a GNN, so for the
simplicial complex extension, $\Sb_0^{(\ell)}$ can be learned
with an MPSN layer \cite{bodnar2021simplicial}.
Using diagonal block matrices of the extended matrix
$\Sb^{(\ell)}$, the reduce step is applied for every
dimension of the complex.
Finally, the connections of the pooled complex are determined
by multiplying cluster assignment matrices against boundary
matrices, for each dimension of the complex.
\nervePool{} (specifically using DiffPool for vertex cluster
assignments) is summarized in Table \ref{Tab:diffPool} in
terms of the necessary select, reduce, and connect steps.
\begin{table}
	\centering
	{
		\renewcommand{\arraystretch}{1.2}\begin{tabular}{lp{0.7\linewidth}}
			\toprule
			{\small\sc Select}                  & $\Sb_0^{(\ell)} =
				\softmax[\MPSN_{\ell,pool}(\Bb_1^{(\ell)},
			\Xb_0^{(\ell)},\Xb_1^{(\ell)})]$                          \\
			\midrule
			{\small\sc Reduce\,($\forall\,p$)}  & $\Xb_p^{(\ell+1)} =
				(\Sb_p^{(\ell)})^T \cdot
				\MPSN_{\ell,embed}(\Bb_p^{(\ell)},\Bb_{p+1}^{(\ell)},
			\Xb_{p-1}^{(\ell)},\Xb_p^{(\ell)},\Xb_{p+1}^{(\ell)})$    \\
			\midrule
			{\small\sc Connect\,($\forall\,p$)} & $\Bb_p^{(\ell+1)} =
			(\Sb_{p-1}^{(\ell)})^T \Bb_p \Sb_p$                       \\
			\bottomrule
		\end{tabular}
	}
	\vspace{1em}
	\caption{DiffPool adjustment for \nervePool{} in the SRC \cite{2021_SRCgraphpooling} framework.}\label{Tab:diffPool}
\end{table}
 Additionally, Fig.~\ref{fig:SX_pooling_diagram} shows an
overview flowchart of \nervePool{}, as implemented via
matrices, when the choice of vertex clusters is based on
DiffPool-style graph pooling.
To extend DiffPool on graphs to simplicial complex
coarsening, we must use separate MPSNs to pool features on
simplices for each dimension and an additional MPSN to find a
soft partition of vertices.
These separate message passing networks (\textit{embedding
	MPSNs} and the \textit{pooling MPSN}) require distinct
learnable parameter matrices $\Theta_{p,embed}^{(\ell)}$ and
$\Theta_{pool}^{(\ell)}$.

For each pooling layer, a collection of \textit{embedding
	MPSNs} take as input the features on simplices (for the
current dimension $\Xb_{p}^{(\ell)}$, a dimension lower
$\Xb_{p-1}^{(\ell)}$, and a dimension higher
$\Xb_{p+1}^{(\ell)}$) and boundary matrices (for the current
dimension $\Bb_{p}^{(\ell)}$ and one dimension higher
$\Bb_{p+1}^{(\ell)}$).
The output of which is an embedding matrix $\Zb_p^{(\ell)}\in
	\R^{n_p^{(\ell)}\times d_p^{(\ell)}}$ for every dim $0\leq p
	\leq \mathcal{P}^{(\ell)}$ of the complex $K^{(\ell)}$.
Here,
$n_p^{(\ell)}$ is the number of $p$-simplices at layer $\ell$,
and $d_p^{(\ell)}$ is the number of $p$-simplex features.
These matrices are the learned $p$-simplex embeddings defined by
\begin{equation*}
	\Zb_p^{(\ell)}=
	\MPSN_{\ell,embed}(\Bb_p^{(\ell)},\Bb_{p+1}^{(\ell)}\Xb_{p-1}^{(\ell)},
	\Xb_{p}^{(\ell)},\Xb_{p+1}^{(\ell)}; \Theta_{p,embed}^{(\ell)}) \, ,
\end{equation*}
where $\Theta_{p,embed}^{(\ell)}$ is the matrix of learnable
weights at layer $\ell$.
Note that for the \textit{pooling MPSN}, which is restricted
to dimension $p=0$, the message passing reduces to a function
\begin{equation*}
	\MPSN_{\ell,pool}(\Bb_1^{(\ell)}, \Xb_0^{(\ell)},\Xb_1^{(\ell)}) \, .
\end{equation*}
The second component of \nervePool{}, illustrated by the
right-side branch in Fig.~\ref{fig:SX_pooling_diagram},
determines the pooling structure through cluster assignments
for each simplex.
For an input simplicial complex, $K^{(\ell)}$, with features
on $p$-simplices,
$\{\Xb_p^{(\ell)}\}_{p=0}^{\mathcal{P}^{(\ell)}}$, we use an
MPSN on the underlying graph of the simplicial complex.
In the same fashion as the DiffPool method for graph pooling
(Equation $5$ in \cite{ying2018diffpool}), we generate an
assignment matrix for vertices (vertex cover) via
\begin{equation*}
	\Sb_0^{(\ell)}=
	\softmax\left[\MPSN_{\ell,pool}(\Bb_1^{(\ell)},
		\Xb_{0}^{(\ell)},\Xb_{1}^{(\ell)}; \Theta_{pool}^{(\ell)})
		\right] \, ,
\end{equation*}
where the output $\Sb_0^{(\ell)}\in\R^{n_0^{(\ell)}\times
	n_0^{(\ell+1)}}$ is the learned vertex cluster assignment
matrix at layer $\ell$.
Note that the MPSN for dimension $p=0$ acting on the
$1$-skeleton of $K^{(\ell)}$, reduces to a standard GNN.
Learning the pooling assignment matrix requires only this
single MPSN, for dimension $p=0$ simplices.
From this single pooling MPSN, $\Sb_0^{(\ell)}$ gives a soft
assignment of each vertex at layer $\ell$ to a cluster of
vertices in the next coarsened layer $\ell +1$.

In terms of network training, the DiffPool learning scheme
for $\Sb_0^{(\ell)}$ uses two additional loss terms added to
the normal supervised loss for the entire complex.
These loss terms are added from each \nervePool{} layer to
the simplicial complex classification loss during training of
the entire network to encourage local pooling of the
simplicial complex and to limit overlapping cluster assignments.
The first term is a link prediction loss, which encourages
clustering of vertices which are spatially nearby on the
simplicial complex,
\begin{equation*}
	\mathcal{L}_{LP} = ||\Ab_0^{(\ell)} , \,
	\Sb_0^{(\ell)}(\Sb_0^{(\ell)})^T ||_F\,,
\end{equation*}
where $||\cdot||_F$ denotes the Frobenius norm.
The second term is a cross entropy loss, which is used to
limit the number of different clusters a single vertex can be
assigned to, given by
\begin{equation*}
	\mathcal{L}_{E} = \frac{1}{n_0^{(\ell)}}
	\sum_{i=1}^{n_0^{(\ell)}} H(\Sb_0^{(\ell)}(i,\cdot))\, ,
\end{equation*}
where $H$ is the entropy function and
$\Sb_0^{(\ell)}(i,\cdot)$ the $i$-th row of the vertex
cluster assignment matrix at layer $\ell$.
Both of the additional loss terms $\mathcal{L}_{LP}$ and
$\mathcal{L}_{E}$ are minimized and for each \nervePool{}
layer, added to the total simplicial complex classification
loss at the end of the network.
Training a \nervePool{} layer reduces to learning a soft
partition of vertices of a graph, since the extension to
coarsening of higher-dimensional simplices are deterministic
functions of the vertex clusters.
\begin{figure}
	\centering
	\includegraphics[width=.9\textwidth]{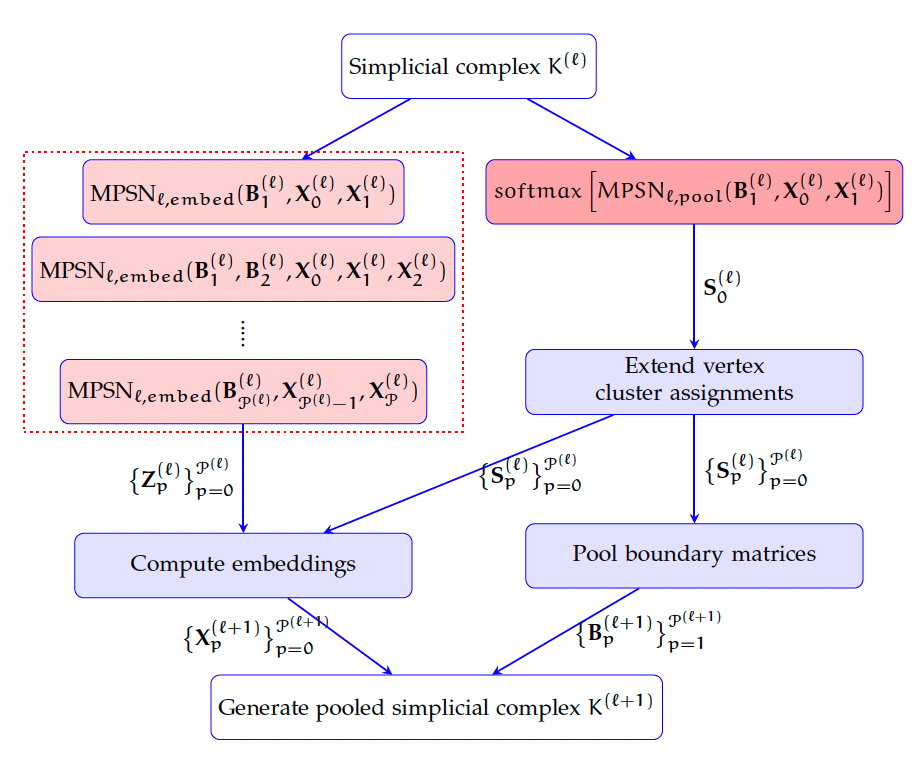}
	\caption{Example \nervePool{} architecture, using DiffPool
		as the motivating vertex cluster method. Takes input
		simplicial complex $K^{(\ell)}$ and returns the pooled
		simplicial complex $K^{(\ell+1)}$. The left-side branch
		uses a collection of MPSNs to compute embeddings for each
		dimension (Reduce step). The right-side branch illustrates
		using an MPSN to compute vertex cluster assignments (Select
		step), and then extending assignments to higher dimensional
		simplices so that this structural information facilitates
		collapsing of simplices when applied against boundary
		matrices (Connect step).}
	\label{fig:SX_pooling_diagram}
\end{figure}

\subsection{Implementation and computational considerations} \label{sec:ImplementationAndComputation}
For our experiments and for the implementation of \NervePool{} we follow the
matrix implementation and use the full cluster assigment matrix. 
Before proceeding with the full implementation we mention that some
modifications are to be made before we can apply \NervePool{} in a machine
learning context.
In graph learning it is often advantageous to \emph{learn} the cluster assignments
from the data through a separate GNN to allow for a data-specific clustering.
Based on the predicted cluster assignments one can then construct the
pooled simplicial complex.
To allow for maximum expressivity we adapt the matrix formulation to allow
for probabilistic cluster assignments, which in the context of machine learning
can be viewed as a weight assigned to the node features when computing the
resulting feature in the pooled complex.

In a machine learning context, the set of simplices of dimension $k$ is often
represented as a tensor of size $(k+1)\times \ell$, where each column represents a
simplex  of dimension containing the indices of vertices spanning the simplex.
Given this representation, the down function (Eqn.~\eqref{eqn:downfunction}) can be viewed
as the maximum operation.
To see this, recall that  the down function is defined as
\begin{equation}
	\Sb_{q,0}[\sigma_q,v_j] :=
	\begin{cases}
		1 & \text{if } \Sb_0[v_m,v_j] = 1 \text{ for any } v_m\in
		\sigma_q                                                  \\
		0 & \text{ otherwise.}
	\end{cases}
\end{equation}
Notice that this is equivalent to
\begin{equation}
	\Sb_{q,0}[\sigma_q,v_j] = \text{max}\left\{ \Sb_0[v_m,v_j]  | v_m\in\sigma_q \right\},
\end{equation}
and this can viewed as rowwise maximum of the rows that span the simplex $\sigma_q$, which is to say that
\begin{equation}
	\Sb_{q,0}[\sigma_q,:] = \text{max}\left\{ \Sb_0[v_m,:]  | v_m\in\sigma_q \right\},
\end{equation}

In graph learning, simplicial complexes are often represented through their
node indices.
We first arbitrarily label vertices $v_i$ with an index and ensure and order
the matrix containing the features such that in both the feature matrix X and
cluster assignment matrix $S$ row $i$ corresponds to vertex $v_i$.
With this convention, the implementation of the down function admits
a particularly fast implementation that can take advantage of hardware
acceleration, since the down function becomes an indexing operation together
with a max operation.

In a similar fashion, the right function can be expressed as a minimum operation
over the columns.
Recall that the right function is defined as
\begin{equation}
	\Sb_{q,p}[\cdot ,\sigma_{p}] := \odot_{v_i \in \sigma_{p}}\Sb_{q,0}[\cdot, v_i].
\end{equation}
The product is $1$ if all entries are $1$ and zero else, which
is equivalent to taking the minimum over the columns and also allows
a particularly fast and parallel implementation.
An additional advantage of the above mentioned modifications is that the
resulting cluster assignment matrix allows the node features to be weighted
according to the learned cluster assignments and enhances the expressivity of
the pooling operation with respect to the features.

\revision{
The above formulation permits a vectorized and hardware accelerable
algorithm for \nervepool{}, which can be used as a machine learning layer for 
pooling simplicial complexes.
Due to its widespread use in machine learning, we implemented
\nervepool{} in \texttt{torch\_geometric}~\cite{fey2019_PyG}.
One drawback of the current implementation is that by default we consider
all possible virtual simplices in each dimension; i.e.~the entire matrix $\Sb$ rather than just the block diagonal.
Given a simplicial complex with $n$ vertices, there are $\binom{n}{k}$ simplices 
in dimension $k$ to consider.
To each virtual simplex we assign a column, which results in the
instantiation of a large matrix when batching is also taken into account.
The right and down functions admit a fast implementation on the GPU,
leading memory to be the current bottleneck.
That said, the current implementation is on par with the other pooling
methods in terms of speed, but there is room for improvement for future work.
}

\revision{
To determine the computational complexity of \nervePool{}, 
let $n_0^{(\ell)}$ be the number of input vertices, let
$C:=n_0^{(\ell+1)}$ be the number of pooled clusters
(meta-vertices), and let $P:=\mathcal{P}^{(\ell)}$ be the
maximal simplex dimension considered by the layer.
For each $q\in\{0,\ldots,P\}$, denote by $n_q^{(\ell)}$ the
number of input $q$-simplices, and define the number of
candidate pooled $p$-simplices by
$	c_p := \binom{C}{p+1}$ for $p=0,\ldots,P$.
If all virtual simplices are considered, the total number of
pooled candidates is
$
	N_{\mathrm{cand}} = \sum_{p=0}^{P} c_p
	$.
The down update (Eq.~\eqref{eqn:downfunction}) computes a
row-wise max over the $(q+1)$ vertices of each $q$-simplex,
giving
\begin{equation*}
	T_{\mathrm{down}} = \sum_{q=1}^{P}
	O\left((q+1)\,n_q^{(\ell)}\,C\right),\qquad
	M_{\mathrm{down}} = \sum_{q=1}^{P}
	O\left(n_q^{(\ell)}\,C\right).
\end{equation*}
The right update (Eq.~\eqref{eqn:rightfunction}) computes
row-wise minima/products across $(p+1)$ columns for each
candidate pooled $p$-simplex:
\begin{equation*}
	T_{\mathrm{right}} = \sum_{q=1}^{P}\sum_{p=1}^{q}
	O\left((p+1)\,n_q^{(\ell)}\,c_p\right),\qquad
	M_{\mathrm{right}} = \sum_{q=1}^{P}\sum_{p=1}^{q}
	O\left(n_q^{(\ell)}\,c_p\right).
\end{equation*}
Hence, in our current implementation, memory is dominated by storing sub-blocks of $\Sb^{(\ell)}$, and scales exponentially in $C$ through $c_p=\binom{C}{p+1}$.
Future work can determine if and how much improvement is potentially possible in terms 
of memory through mitigation strategies such as limiting maximal simplex dimension, restricting candidate simplices to those present in the original complex, using sparsity patterns from boundary matrices, or sampling-based approximations for the right-update. 
}

\section{Experiments} \label{sec:Experiments}

Given our general framework for pooling for simplicial complexes, we
asses our method in a comprehensive suite of experiments.
\nervepool{} is evaluated on a graph datasets, as all comparison partners
are restricted to graph level tasks, in particular the task of
graph classification.
\revision{
To ensure that our method also works for higher order simplicial
complexes, we evaluate NervePool on a topological
dataset~\cite{Ballester25a} of $2$-manifolds. 
}
We are particularly interested in the capacity of \nervepool{} to
capture and summarize the global structure of the graphs and meshes.

\subsection{Preprocessing and experimental setup.}
Following the experimental section of DiffPool~\cite{ying2018diffpool},
we evaluate our method on several benchmark datasets~\cite{morris2020tu}
that have graph classification as the target.
The graphs in the datasets are filtered to have a fixed maximum number of
nodes, which is set to $700$ for Proteins, $500$ for DD and $150$ for
the other datasets.
We use a $70\%/10\%/20\%$ train/validation/test split.
As a loss term, we use \emph{categorical cross entropy} for graph
classification and the Adam optimizer~\cite{kingma2015adam} with a learning rate of $0.001$ for all
models. We train all models for $50$ epochs and repeat said training
with $5$ different seeds, reporting the mean and standard deviation of
the classification accuracy.

\subsection{Architectures.}
The base architecture for all models consists of three GCN layers with
$64$ hidden channels interleaved with two graph pooling layers.
To ensure a fair comparison between the different pooling methods, we fix
the base architecture and only change the pooling method; while further
changes might yield to slightly improved results, our current setup
demonstrates that \nervePool{} provides a stable pooling layer that can
plugged into any architecture.
As both \nervepool{} and DiffPool require a learned cluster assignment for
the pooling, we use an embedding GNN for each pooling layer that predicts
the node clustering assignments.
As final aggregation of the graph representation we calculate the mean
across the hidden features and the final prediction is done with
a $2$-layer MLP with ReLU activation.

\subsection{Classifying Geometric Graphs.}
The results are reported in
Table~\ref{tab:experiment_graph_classification}. We find \modelname{} to
perform on a par with other pooling methods. In the first column, we
show the average rank across the four data sets, and see that while
DiffPool performs best, \NervePool{} is a consistently close second.
This is particularly notable given that \modelname{} is a more general
method, applicable to higher-order simplicial complexes, while the other
methods are restricted to graphs.
Most notably, \modelname{} incorporates the edge features in a natural
and mathematically grounded fashion through the $S$ matrix at each layer.
Since the graphs in our datasets do not contain edge features, we set
them to the constant zero vector.
This means that the edge features do not contribute to the clustering
in this case, but they are still pooled and passed to the next layer.
In future work it would be interesting to see if the inclusion of edge
data in this setup will yield a more significant advantage for
\nervepool{}.

For future work, we believe it would be interesting to check if the
inclusion of edge data in this setup will yield a more significant
advantage for \nervepool{}. We note also that very little is needed for
\nervePool{} in the way of user-defined inputs compared to the other
methods; in particular, it does not require specifying a pooling ratio
in advance. Users only have to choose an architecture into which
\nervePool{} should be integrated as well as a loss term for training
the model, i.e., standard requirements for any graph learning model.
This, coupled with the fact that \nervePool{} essentially provides
\emph{probabilities}, makes it a conceptual extension and generalization
of DiffPool.

\begin{table}
	\centering
	\sisetup{
		detect-all              = true,
		table-format            = 2.2(2),
		detect-mode             = true,
		separate-uncertainty    = true,
		retain-zero-uncertainty = true,
		table-align-text-after  = false,
		mode                    = text,
	}\begin{tabular}{l S SS SS}
	\toprule
	{\small\sc Model} & {Average Rank} & {\small\sc DD} & {\small\sc Enzymes} & {\small\sc MUTAG} & {\small\sc PROTEINS} \\
	\midrule

	DiffPool          & 1.00           & 0.78 \pm 0.02  & 0.38 \pm 0.06       & 0.84 \pm 0.07     & 0.75 \pm 0.02        \\
	NoPool            & 3.25           & 0.71 \pm 0.04  & 0.32 \pm 0.05       & 0.72 \pm 0.06     & 0.72 \pm 0.04        \\
	SAGPool           & 3.50           & 0.69 \pm 0.02  & 0.29 \pm 0.06       & 0.8 \pm 0.03      & 0.71 \pm 0.04        \\
	TopKPool          & 5.00           & 0.69 \pm 0.06  & 0.28 \pm 0.06       & 0.69 \pm 0.09     & 0.7 \pm 0.04         \\
	\midrule
	NervePool         & 2.25           & 0.73 \pm 0.05  & 0.33 \pm 0.05       & 0.77 \pm 0.19     & 0.75 \pm 0.02        \\
	\bottomrule
\end{tabular}
 	\vspace{1em}
	\caption{Graph classification results for the TU Datasets were the average
		accuracy over $5$ runs is reported.
As model performance varies across the datasets, we rank the models
		from $1$ (best performance) to $5$ (worst performance) and compute
		the average rank across the datasets.
\nervePool{} consistently ranks second, showing a good average
		performance.
	}
	\label{tab:experiment_graph_classification}
\end{table}

\subsection{Classification of manifolds}
\revision{
To highlight the fact that \nervePool{} also works well with higher
order datasets, we evaluate our pooling method on the
MANTRA~\cite{Ballester25a} dataset.
The dataset consists of the full enumeration of $2$-manifolds with up to 
$10$ nodes and contains $21$ types of manifolds.
Following the methodology for the graph learning benchmark, we evaluate 
the same models for the classification of manifolds. 
We keep the architecture and preprocessing the same, with the difference
that \nervePool{} is allowed to operate on the full simplicial complex,
rather than just the $1$-skeleton as required by the rest of the pooling methods tested. 
By definition the dataset does not include node features and to ensure
the models learn the topological and combinatorial structure, we add
random node features to each vertex.
The dataset has a total of up to $10$ nodes, with the vast majority
falling in this category;
depending on the pooling method, we  pool the complex down to $5$ nodes when a hard threshold is required, or 50\% of the input vertices when a ratio is required.
The results for the manifold classification task are presented in
Table~\ref{tab:experiment_mantra} and show \nervePool{} to outperform
all other methods. 
The topological nature of the task suggests that \nervePool{} preserves
the important topological features of the manifold for later downstream
classification, whereas the graph methods lose this structure.
}
\begin{table}[h]
	\centering
	\sisetup{
		detect-all              = true,
		table-format            = 2.2(2),
		detect-mode             = true,
		separate-uncertainty    = true,
		retain-zero-uncertainty = true,
		table-align-text-after  = false,
		mode                    = text,
	}\begin{tabular}{l Sc}
	\toprule
	{\small\sc Model} & {\small\sc MANTRA} & {\small Rank}  \\ 
	\midrule
  DiffPool & 0.47 \pm 0.01 & 3 \\
  NoPool & 0.46 \pm 0.01 & 5 \\
  SAGPool & 0.89 \pm 0.01 & 2 \\
TopK & 0.34 \pm 0.01 & 6\\
  \midrule
  NervePool & 0.92 \pm 0.01 & 1 \\
  \bottomrule
\end{tabular}

 	\vspace{1em}
	\caption{\revision{
    Classification results on the MANTRA dataset of $2$-manifolds with
    random node features; we use only graph input information for the first four methods, while \NervePool{} uses higher dimensional simplices. 
In this experiment,\nervePool{} shows superior performance and outperforms all other methods. 
Learning the clustering allows our method to preserve important
    properties of the simplicial complex and thus allows better expressivity. 
}
  }
	\label{tab:experiment_mantra}
\end{table}
 \section{Conclusions and future directions}
\label{sec:conclusions}
In this paper, we have defined a method to extend learned
graph pooling methods to simplicial complexes.
Given a learned partition of vertices, \nervePool{} returns a
coarsened and simplified simplicial complex, where simplices
and signals defined on those simplices are redistributed on
the output complex.
This framework for simplicial complex coarsening is a very
flexible method because the input partition of the vertex set
can come from any choice of standard graph pooling method or
learned vertex clustering.
We show that there is a choice of input cover on the vertices
such that \nervePool{} returns the same simplicial complex
(up to re-weighting) and that when used in the context of a
simplicial neural network with hard vertex clusters, it is a
simplex-permutation invariant layer.
Additionally, we prove the equivalence of the nerve/cover
topological interpretation and matrix implementation using
boundary matrices for the setting restricted to hard vertex partitions.
This pooling layer has potential applications in a range of
deep learning tasks such as classification and link
prediction, in settings where the input data can be naturally
modeled as a simplicial complex.
\nervePool{} can help to mitigate the additional computation
cost of including higher dimensional simplices when using
simplicial neural networks:
reducing complex dimension, while redistributing information
defined on the simplicial complex in a way that is optimized
for the given learning task.

A limitation of this method is that using standard graph
pooling methods for the initial clustering of vertices limits
the learned influence of higher-dimensional simplices.
Future work in this direction to include more topological
information in the learning of vertex clusters could enhance
the utility of this method for topologically motivated tasks.
For graph pooling methods with auxiliary loss terms such as
DiffPool \cite{ying2018diffpool} and MinCutPool
\cite{bianchi2020spectral}, adjustment or inclusion of
additional auxiliary topological loss terms to encourage
vertex clusters with topological meaning would be an
interesting line of inquiry.
Taking into account the higher-dimensional structure for the
learned pooling on underlying graphs could allow \nervePool{}
to coarsen simplicial complexes, tuned such that specified
topological structure is retained from the original complex.

While the current assumptions of \nervePool{} include loss
functions only defined on the $1$-skeleton, if additional
auxiliary loss terms defined on higher-order simplices were
utilized, then the functions defining \nervePool{}
would need to be continuous and differentiable for gradient
computations.
This is automatic for Equation \ref{eqn:rightfunction}, the
right arrow update, since it defines a series of Hadamard
(entry-wise) products, which are differentiable.
However, Equation \ref{eqn:downfunction} is an ``on or off''
function, which is not differentiable~(while rewriting this
is possible, it still is a hard vertex assignment).
To address this, we could define the relaxation of the down function 
to be the maximum of the entries, which in the case of binary 
entries agrees with the theoretical definition.

In this paper, to compute embeddings of features on
$p$-simplices
$\{\Xb_p^{(\ell+1)}\}_{p=0}^{\mathcal{P}^{(\ell+1)}}$, we
choose to apply separate coarsening for each dimension of the
input simplicial complex in Equation \ref{eqn:poolingX}.
This choice enforces that information sharing for the signals
defined on simplices only affect signals on pooled simplices
within a given dimension.
Alternatively, future modifications could utilize the entire
block matrix $\Sb^{(\ell)}$ applied against a matrix
containing signals on simplices of all dimensions, which
would allow for signal information to contribute to pooled
simplices amongst different dimensions.

\revision{
We have also restricted our theoretical results to hard partitions of the vertex set, but everything can be defined and used in the case of soft partitions. 
It would be interesting in future work to explore the theoretical properties of \nervePool{} in the case of soft partitions, if the matrix and topological formulations are equivalent in this case (as in Thm.~\ref{thm:topEquivMat}), and to see if there are any advantages to using soft partitions in practice.
Another open direction for work would be to understand if improvement could be made by enforcing the input partitions in such a way as to return good covers for \nervePool{} and thus allow for the application of the nerve lemma. 
}
\revision{
We also note that all work in this paper forgot any potential orientation of simplices. 
There is room for improvement to see if tracking the orientation of simplices in the pooling process would still satisfy the theoretical guarantees, and/or if it would yield better results in practice in cases where orientation of simplices matters.
}

Moving from a theoretical framework to an implementation 
requires certain architectural decisions and further 
investigating different choices could provide interesting 
insights into the capacity of \nervePool{}.
\revision{
Relaxing the down and right functions further to use the product 
and sum instead of minimum and maximum could potentially improve gradient flow 
to the pooling networks. 
}
Additional future work on simplicial complex coarsening
should include experimental results to identify cases where
including simplicial pooling layers improve, for example,
classification accuracy.
Also, comparison experiments at different layers of a
simplicial neural network would be useful to determine if the
pooled complexes are reasonable and/or meaningful
representations of the original simplicial complexes, with
respect to a given task.

\paragraph*{Funding}
This work was funded in part by the National Science Foundation, through grants CCF-1907591, CCF-2106578, and CCF-2142713 as well as NIH NIGMS-R01GM135929. This work has received funding from the Swiss State Secretariat for Education, Research, and Innovation (SERI).

\paragraph*{Acknowledgements}
Sarah McGuire and Ernst R\"oell should be considered as co-first authors. The authors would like to thank Sourabh Palande for helpful discussions.

\bibliography{references}
\end{document}